\documentclass[11pt]{article}
\usepackage{amsfonts}
\usepackage{eurosym}
\usepackage{amsmath} 
\usepackage{amssymb}
\usepackage{latexsym}
\usepackage{graphicx}
\usepackage{amsthm}
\usepackage{hyperref}
\usepackage{color}
\usepackage{bbm}
\usepackage{float} 
\usepackage[normalem]{ulem}
 \usepackage{url}

\usepackage[a4paper,top=3cm,bottom=2cm,left=3cm,right=3cm,marginparwidth=1.75cm]{geometry}

\newcommand{\ignore}[1]{}

\newtheorem{definition}{\hspace{0pt}\bf Definition}

\newtheorem{lemma}{\hspace{0pt}\bf Lemma}
\newtheorem{theorem}{\hspace{0pt}\bf Theorem}

\newcommand{\x}{\tilde x}
\newcommand{\s}{\tilde s}
\newcommand{\iu}{\tilde i_u}
\newcommand{\id}{\tilde i_d}
\newcommand{\ru}{\tilde r_u}

\title{Optimal adaptive testing for epidemic control: combining molecular and serology tests }

\usepackage{authblk}

\author[1]{\quad  D. Acemoglu}
\author[2]{A. Fallah}
\author[3]{A. Giometto}
\author[2]{D. Huttenlocher}
\author[2]{A. Ozdaglar}
\author[4]{\qquad F. Parise\thanks{Corresponding author: F. Parise (fp264@cornell.edu), D. Acemoglu (daron@mit.edu),  A. Fallah (afallah@mit.edu), A. Giometto (giometto@cornell.edu), D. Huttenlocher (huttenlocher@mit.edu), A. Ozdaglar (asuman@mit.edu),  S. Pattathil (sarathp@mit.edu). We acknowledge support from the C3.AI grant.\\ Authors are listed in alphabetical order.
}}
\author[2]{S. Pattathil}
\affil[1]{\small Department of Economics, Massachusetts Institute of Technology (MIT)}
\affil[2]{Department of Electrical Engineering and Computer Science, MIT}
\affil[3]{School of Civil and Environmental Engineering, Cornell University}
\affil[4]{Department of Electrical and Computer Engineering, Cornell University}

\date{\vspace{-1.5cm}}
\begin{document}
\maketitle

\begin{abstract}
The COVID-19 crisis highlighted the importance of non-medical interventions, such as testing and isolation of infected individuals,  in the control of epidemics. Here, we show how to minimize testing needs while maintaining the number of infected individuals below a desired threshold. We find that the optimal policy is adaptive, with testing rates that depend on the epidemic state. Additionally, we show that such epidemic state  is difficult to infer with molecular tests alone, which are highly sensitive but have a short detectability window. Instead, we propose the use of baseline serology testing, which is less sensitive but detects past infections, for the purpose of state estimation. Validation of such combined testing approach with a stochastic model of epidemics shows significant cost savings compared to non-adaptive testing strategies that are the current standard for COVID-19.
\end{abstract}

\section{Introduction}

A large literature in mathematical epidemiology has studied how to control
and eradicate diseases by means of therapeutics and vaccinations, \cite%
{nowzari2016analysis,behncke2000optimal}. However, the influenza pandemic of
1918 and the  current COVID-19 pandemic underscore the difficulty of such
eradication in the case of virulent viruses,\ and have necessitated measures
to reduce transmissions, for example with the use of face masks \cite{chu2020physical}, social distancing and costly
lockdown measures \cite{flaxman2020estimating,bertuzzo2020geography} \cite{di2020impact}. Another powerful tool to limit
transmissions is early identification of infected individuals and  epidemic hot-spots in local communities, which can both be accomplished by
testing \cite{grassly2020report,oecd2020report}. Nevertheless, during the COVID-19 pandemic testing resources have proven to be limited and
expensive in much of the world \cite{AACC,NYT2,NYT3,pullanounderdetection}; in the US, lack of
testing capacity not only helped spread the virus but also led to the
underestimation of the severity of the pandemic in the first half of 2020, 
\cite{NYT}. This crucial role of testing notwithstanding, the question of how limited
testing resources can be deployed to optimally control the spread of a pandemic has
attracted relatively little systematic attention.

In this paper, we derive an optimal (dynamic) testing strategy in an SIR
(Susceptible, Infected, Recovered) model of epidemics. Because undetected
individuals may pass the disease to others and may be more likely to develop
serious symptoms requiring hospitalization, we start by assuming that the
number of undetected infected individuals has to be kept below a maximum $%
i_{\max }$ at all times. We show that the optimal testing strategy takes a
simple form: the testing rate has to be time-varying in order to satisfy
the constraint, and takes the form of a \emph{most
rapid approach path}, \cite{spence1975most}. Namely, there is no testing
until undetected infections reach $i_{\max }$, after which testing resources
are used to keep infections at the threshold $i_{\max }$ until infections
decline naturally, bringing the pandemic to an effective close. The
intuition for this result is that it is not worth using testing resources to
keep undetected infections strictly below $i_{\max }$ so long as the
pandemic is still ongoing and  infections cannot be brought down to
zero. Hence, the best approach is to let the infection reach the threshold
and then keep it there with a time-varying testing
policy. Note that such optimal time-varying strategy is state-dependent,
that is, the level of testing is dependent on the epidemic state (the current number of infected, susceptible and recovered individuals).

The second contribution of our paper starts by recognizing that the epidemic state needed to implement the 
aforementioned optimal testing strategy is typically hard to know precisely, as highlighted by the early stages of  COVID-19 spread  in the US.
In fact, the most common qPCR tests for COVID-19, which are {molecular tests}
based on  detection of the virus' genetic material via quantitative polymerase chain reaction, may be ill-suited to obtain
such aggregate information. These tests identify infected individuals only
during a short window of time. For example, according to \cite%
{kucirka2020variation} the probability of COVID-19 detection via qPCR is
above 75\% in a window of roughly a week within active cases, while \cite%
{roche} gives a three weeks window for detectability of cases via qPCR.
  In contrast, {serology tests}, which detect antibodies produced by the immune system in response to current and past infections, identify infections during a longer window \cite{kubina2020molecular}, but are typically less sensitive and thus received relatively less attention in the epidemic control  literature.
These problems are unlikely to be confined to COVID-19 and would probably
recur in the event of future pandemics. 
To systematize these observations we study the effectiveness of both types of tests for identifying the epidemic state and find that \textit{serology tests}, which have lower sensitivity but are  faster, cheaper and can reveal past infections,  offer a better alternative than 
qPCR-like tests  (from here on termed \textit{molecular tests}), which have high accuracy but a short-window of detection. Namely, we show that if the transmission rate  in the SIR
model is time-varying, then the epidemic state cannot be identified---is
not observable---with  molecular tests alone. Intuitively, just observing the flow
of current infections may not be sufficient to distinguish the nonlinear
dynamic evolution of the system due to different initial conditions versus
different time-varying trajectories of transmission rates. Serology testing overturns
this result, however, by providing a cheap way of estimating the stocks of past infected and recovered individuals. 
This is despite the fact that serology tests may
have significant Type II errors because of  low sensitivity (especially
in the first stage - 0 to 6 days - of infection, \cite{roche}). Indeed, we find that
Type II errors, which can be very costly when the purpose is  to diagnose 
individual infections, are not problematic for the purpose of 
estimating the epidemic state (which is  aggregate information). For such purpose,  detection of recovered individuals is a more important feature than high sensitivity. Hence,  serology
tests are an ideal complement to molecular tests for the purpose of estimating
aggregate infections---an intuition that is formally established by our
mathematical analysis.

In addition to our main analysis, we consider two important extensions.
First, we study a variant of our baseline model of optimal testing in which
both undetected and detected infections have to be kept below a maximum
threshold. In this case, the optimal testing strategy is more complex. Nevertheless, we establish that the basic insights from
our baseline analysis generalize to this problem. Second, we recognize that
in reality the dynamics of epidemics are intrinsically stochastic. To
confirm the robustness of the proposed approach, we apply our testing methodology to
a stochastic continuous-time Markov chain model of epidemics and develop an extended Kalman filter \cite{khalilbook} to
estimate the epidemic state in the presence of stochasticity.  To this end, we exploit an expansion  of the master equation governing the probability distribution of the Markov chain model \cite{vankampen,gardiner2009stochastic} to derive  a description of the epidemic dynamics in terms of a Langevin equation, whose mean coincides with the deterministic SIR model used for computing the optimal testing strategy. The covariance matrix of the noise term in the  Langevin   equation (which can be explicitly characterized as a function of the epidemic state)  is then given as input to the extended Kalman filter  to optimally incorporate new observations in the model predictions.

Our paper is related to the growing literature on SIR models, especially
applied to the recent COVID-19 pandemic. Classic references on the SIR model
and its applications to model epidemics include \cite%
{kermack1927contribution, daley2001epidemic, diekmann2001mathematical,
keeling2011modeling, andersson2012stochastic}. Additionally, see \cite%
{pastor2015epidemic} for a review of models of epidemic processes over
networks and \cite{anderson1992infectious, nowzari2016analysis} for analysis
and control of epidemic models. Several papers developed more general
compartmental models for analyzing the spread of COVID-19 and examining the
effects of interventions (see e.g., \cite{gatto2020spread, atkeson2020will,
stock2020data, cornell,zhang2020changes,chinazzi2020effect}). Other papers, such as \cite%
{pare2017epidemic, hota2019game, 8488684, hota2020generalized} analyzed the
spread of epidemics over both time-varying and static networks.

Our work is more closely connected to a smaller literature that considers
testing within this framework. \cite{alvarez2020simple, acemoglu2020optimal,brotherhood2020economic} consider testing in the context of
optimal lockdown policies in SIR models (in the latter two papers with
explicit recognition of heterogeneities across different age groups).
Neither of these papers studies optimal testing, nor discusses the problem
of identifying the underlying epidemic state, which is assumed to be known
in  this branch of the literature. \cite{acemoglu2020testing}
considers optimal testing in a simple model of disease percolation, but
their focus is on the countervailing effects that testing creates by discouraging social distancing among certain groups of individuals. Moreover, their analysis is
simplified by focusing on a non-SIR percolation model that enables explicit
characterization and they do not discuss the issue of estimating the
underlying epidemic state. \cite{drakopoulos2020demand} studies settings
where accurate tests are not available in abundance. They show that
moderately good tests provide enough information to have a positive social
outcome, and that it is not optimal to wait for tests with very high
accuracy. Similarly, \cite{larremore2020test} compared molecular
and antigen tests and found that test sensitivity is less important than testing
frequency for screening purposes. \cite{kraay2020modeling} suggests the use
of serology testing to allow seropositive individuals (i.e., individuals
with immunity) to increase their level of social interaction. They conclude
through extensive simulations that serology testing has the potential to
mitigate the impacts of the COVID-19 pandemic while also allowing a
substantial number of individuals to safely return to social interactions
and to the workplace. Information about the epidemic state of individuals obtained through both qPCR and serology tests is used in \cite{li2020disease} to derive disease-dependent lockdown policies.  \cite{vespignani2020modelling} highlights the need for integrating seroepidemiological data into transmission models to  reduce the uncertainty in the parameter estimates of clinical severity and transmission dynamics. \cite {behncke2000optimal} studies the optimal testing policy when the objective is a weighted combination of the cost of infection and testing with no constraints on the state variables. They show that the optimal testing policy takes the maximal value until some time and then zero after. 
The key novelty of our work is to suggest the combined use of
serology testing for the purpose of state estimation with molecular testing
for optimal containment of the number of infected within a desired
thresholds, together with an analytic derivation of the optimal
adaptive testing rates.

In characterizing optimal controls in an SIR framework, our paper is related
to a few other papers that study optimal lockdown policies in SIR models.
These include \cite{miclo2020optimal}, which provides an analytical
characterization of optimal lockdown policies in a setting where suppression
is costly and there is an upper bound on the number of infections,
representing a constraint on intensive care unit resources. Another related paper in this
regard is \cite{kruse2020optimal} which studies optimal social distancing
measures to minimize a combination of the total health and economic cost of
the infected population and the cost of reducing the transmission rate. The key difference of our
work is that we focus on testing as a means to identify and isolate infected
individuals instead of lockdown policies that impose a degree of isolation
to an entire community.

Finally, because of its analytical focus, our paper is distinguished from a
large number of recent papers that analyze intervention policies numerically
(see e.g., \cite{zaman2008stability, sharma2015stability, di2017optimal,
farboodi2020internal},  \cite{ gollier2020group}, \cite{berger2020seir} among others).
As detailed in the  discussion section, we believe that the
theoretical insights generated from our analysis of the SIR model highlight
fundamental mechanisms and properties related to the use of testing as a
tool for controlling an epidemic, that can then be generalized and refined to
more sophisticated models of specific epidemics as typically done in the
numerical literature above.

\section{Results}

\subsection{Optimal adaptive testing strategies}

\label{sec:opt}

We model the progression of the epidemic in a population via a
Susceptible-Infected-Recovered (SIR) model with three compartments corresponding to susceptible, infected and recovered individuals. Testing is introduced in the
model by partitioning the infected compartment into infected individuals
that have not yet been detected and are free to circulate, which we term
infected-undetected  (I$_\textup{u}$), and infected individuals that have been
detected and are therefore separated from the general population (e.g.,
quarantined), which we term infected-detected (I$_\textup{d}$) (see Fig. \ref{fig:sketch_SIR}A). 
Infected-undetected
individuals infect susceptible ones with possibly time-varying transmission rate $\beta(t)$, denoting the number of
contacts per unit time multiplied by the probability that a contact leads to
infection. For the purpose of controlling the epidemics, we assume that molecular testing (such as qPCR) is performed at rate $\theta(t)$ and has
sensitivity $\eta$, i.e., infected individuals who get tested are
detected with probability $\eta$. Individuals that test positive for the
infection are quarantined and moved to the infected-detected compartment. Infected individuals become recovered  with rate $\gamma$. The corresponding model
equations are: 
\begin{equation}  \label{eq:SIR_optimal}
\begin{aligned}
\frac{ds(t)}{dt} &= -\beta(t) s(t) i_u(t) \\
\frac{di_u(t)}{dt} &= \beta(t) s(t) i_u(t) - \gamma i_u(t) - \eta
\theta(t) i_u(t) \\
\frac{di_d(t)}{dt} &= \eta \theta(t) i_u(t) -
\gamma i_d(t)\\
\frac{dr(t)}{dt} &= \gamma i_u(t) + \gamma i_d(t),
\end{aligned}
\end{equation}
where small letters denote the fraction of individuals in each compartment.
We assume that the population size is constant with time, thus the last
equation is redundant since $r=1-s-i_u-i_d$.

\begin{figure}[!h]
\centering
\includegraphics{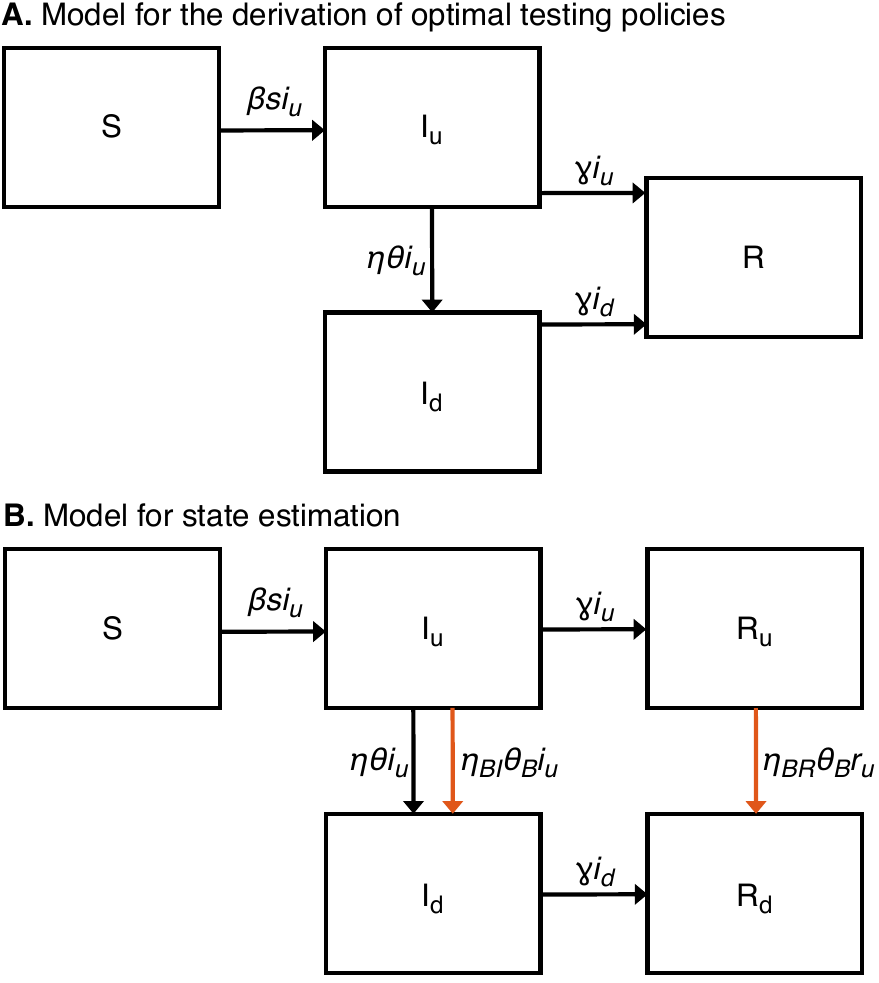}
\caption{The deterministic epidemiological models used for the derivation of
optimal testing policies (A) as discussed in Section \ref{sec:opt} and for state estimation (B) as discussed in Section \ref{sec:state_est}. Individuals in
the population are divided in compartments according to their
epidemiological state. In (A), Susceptible individuals (S) are infected via
contact with infected-undetected individuals ($\textup{I}_\textup{u}$), which either
transition to the infected-detected ($\textup{I}_\textup{d}$) compartment if they are tested
positive for infection, or to the Recovered (R) compartment if they
recover from the infection before being tested positive. Infected-detected
individuals also recover from the infection. In (B), infected-undetected
individuals ($\textup{I}_\textup{u}$) can transition to the infected-detected state by being
tested positive via either adaptive molecular testing or baseline testing.
Infected-undetected individuals can also transition to the recovered-undetected ($\textup{R}_\textup{u}$) compartment, if they recover from the infection before
being tested positive or displaying symptoms. With serology baseline
testing, recovered-undetected individuals can transition to the recovered-detected compartment ($\textup{R}_\textup{d}$) if they are tested positive for past infection.
Infected-detected individuals transition directly to the recovered-detected
compartment when they recover from the infection. Transitions between
compartments are indicated along with the corresponding rates, orange arrows
indicate transitions due to serology testing.}
\label{fig:sketch_SIR}
\end{figure}

Subject to these dynamics, we aim at solving the
following constrained optimization problem:
\begin{equation}
\begin{aligned} \min_{\theta(\cdot)\ge 0} &\quad
\int_{t=0}^{\infty}\theta(t) dt \\ \text{such that:}& \quad i_u(t) \leq
i_{\max} \quad \forall t . \end{aligned}  \label{eq_reduced_problem}
\end{equation}%
In words, our goal is to design the optimal adaptive testing rate to minimize
the total number of tests needed while controlling the epidemic so that the
fraction of infected-undetected individuals always remains below a desired threshold $%
i_{\max }\ll 1$. This constraint is motivated by two considerations. First,
infected-undetected individuals circulate freely in the society and infect
others, and thus a high number of such individuals would lead to a rapid
takeoff in infections. Second, because they do not receive care,
undetected infected individuals may later develop severe complications, and may need
emergency intensive care unit (ICU) capacity, which has proven to be in short supply during the COVID-19 pandemic. The
appropriate level of $i_{\max }$ is a policy choice, and depends on
several factors, including whether policymakers are intending to keep the
reproduction rate of the pandemic below one and the maximum surge capacity
of ICU resources.

Our first main result provides a complete characterization of the optimal adaptive policy for problem \eqref{eq_reduced_problem}.
For simplicity we here discuss the case when the transmission rate $\beta$ is constant (i.e., $\beta(t)=\beta$ for all times). In this case, we say that the system reaches \textit{herd immunity} when the
epidemic state is such that $s(t)={\gamma}/{\beta}$, as from that time on
the number of infected-undetected individuals decreases even without testing. We prove an analogous theorem with time-varying monotonic $%
\beta (t)$ in the Supplementary Materials.

\begin{theorem}
\label{optimal} The optimal testing policy $\theta^{\dagger}(t)$ for problem \eqref{eq_reduced_problem} with dynamics as in  Eq.\ \eqref{eq:SIR_optimal} and constant transmission
rate $\beta$ is described in three phases:

\begin{enumerate}
\item While $i_u(t)<i_{\max}$, do not test, i.e., set $\theta^\dagger(t)=0$.

\item After $i_u(t)$ reaches $i_{\max}$, test with time-varying rate $%
\theta^\dagger(t)=({\beta s(t)-\gamma})/{\eta}$.

\item Once herd immunity is reached, stop testing, i.e., set $%
\theta^\dagger(t)=0$.
\end{enumerate}
\end{theorem}

\begin{figure}[h]
\centering
\includegraphics{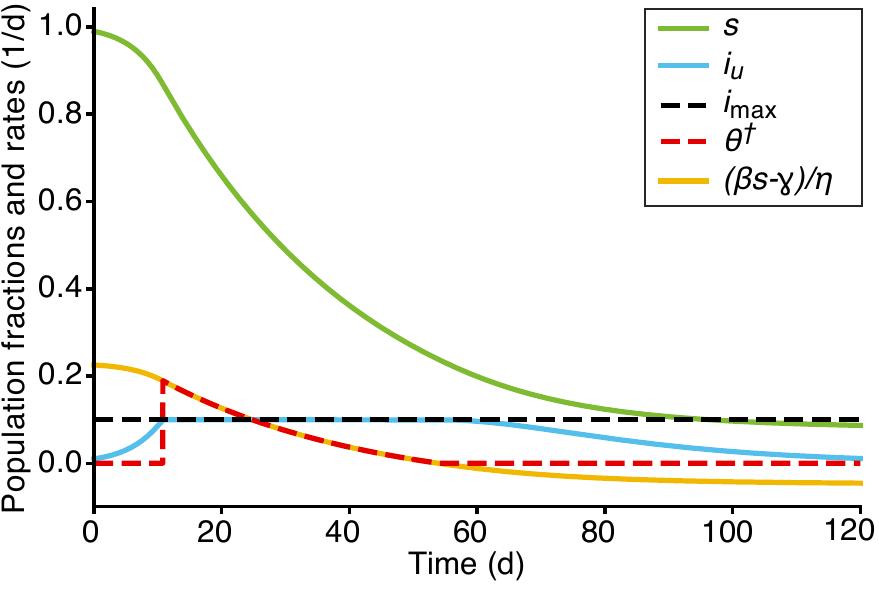}
\caption{Illustration of  the optimal testing policy for problem
\eqref{eq_reduced_problem}. The green and light-blue curves are, respectively, the fractions of
susceptible and infected-undetected individuals in the population. The
fraction of infected-undetected $i_u$ is kept below the constraint $i_{\max}$
(black, dashed line) at all times by the optimal adaptive testing
policy (red, dashed curve), which is equal to zero until $i_u$ reaches $%
i_{\max}$, is then equal to $(\protect\beta s(t)-\protect\gamma)/\protect%
\eta $ (yellow curve) until herd immunity is reached (i.e., when $s(t)=\protect\gamma/%
\protect\beta$), and is equal to zero afterwards.  For illustration purposes, we set $i_{\max}=0.1$ and $\eta=1$, the other parameters are as in Table \ref{table:parameters}.}
\label{fig:Optimal_SIR}
\end{figure}

As illustrated in Fig. \ref{fig:Optimal_SIR}, the optimal policy for problem  \eqref{eq_reduced_problem} starts testing only
once the constraint on $i_{u}$ is attained, and then sets a time-varying
rate $({\beta s(t)-\gamma})/{\eta}$ such that $d{i}_{u}(t)/dt=0$, keeping the fraction of
infected-undetected individuals constant at the threshold $i_{\max }$ until herd
immunity is reached, after which there is no need for further testing as the
epidemic tends naturally towards extinction.  Intuitively, given that no testing is needed once herd immunity is reached, the optimal policy takes the form of the most rapid approach path introduced in \cite{spence1975most} to reach herd immunity as fast as possible, while satisfying the $i_{\max}$ constraint.  The testing policy detailed above leads to the highest  number of infected-undetected individuals by employing no testing until undetected infections reach $i_{\max}$ and then utilizing testing resources  to keep the infections at this threshold, thus guaranteeing the most rapid feasible path to  herd immunity.

\subsection{State Estimation}\label{sec:state_est}

The optimal adaptive testing policy $\theta ^{\dagger }(t)$ derived in
the previous section depends on knowledge of the {\it aggregate epidemic state}, i.e., the values of $s$, $i_{u}$, $i_{d}$ and $r$ at all times. In practice, because this information is not readily available, the state of
the epidemic must be estimated from detected infections. Additionally, the policy makers must  know the model parameters. While many of such parameters are related to properties of the disease,  the  transmission rate is a function of people's behavior \cite{weitz2020awareness}, is typically time-varying and needs to be estimated from data.
This is a nontrivial problem, since the dynamics induced by the SIR model is
highly nonlinear and a given path of infections can be due to different
$\beta (t)$  trajectories coupled with different initial conditions.

To address these problems, we propose the use of \textit{%
baseline testing} with a constant rate $\theta _{B}$ to complement the 
\textit{adaptive testing} policy derived above, see Fig.~\ref{fig:scheme}.
Importantly, the objective of baseline testing is not to
control the epidemic, but rather to collect enough data to robustly estimate
the state of the epidemic and the parameter $\beta(t)$. The policy maker can
then use the estimated state and parameter to implement the optimal
adaptive policy discussed in Section \ref{sec:opt}.

\begin{figure}[h]
\centering
\includegraphics[width=3.5in]{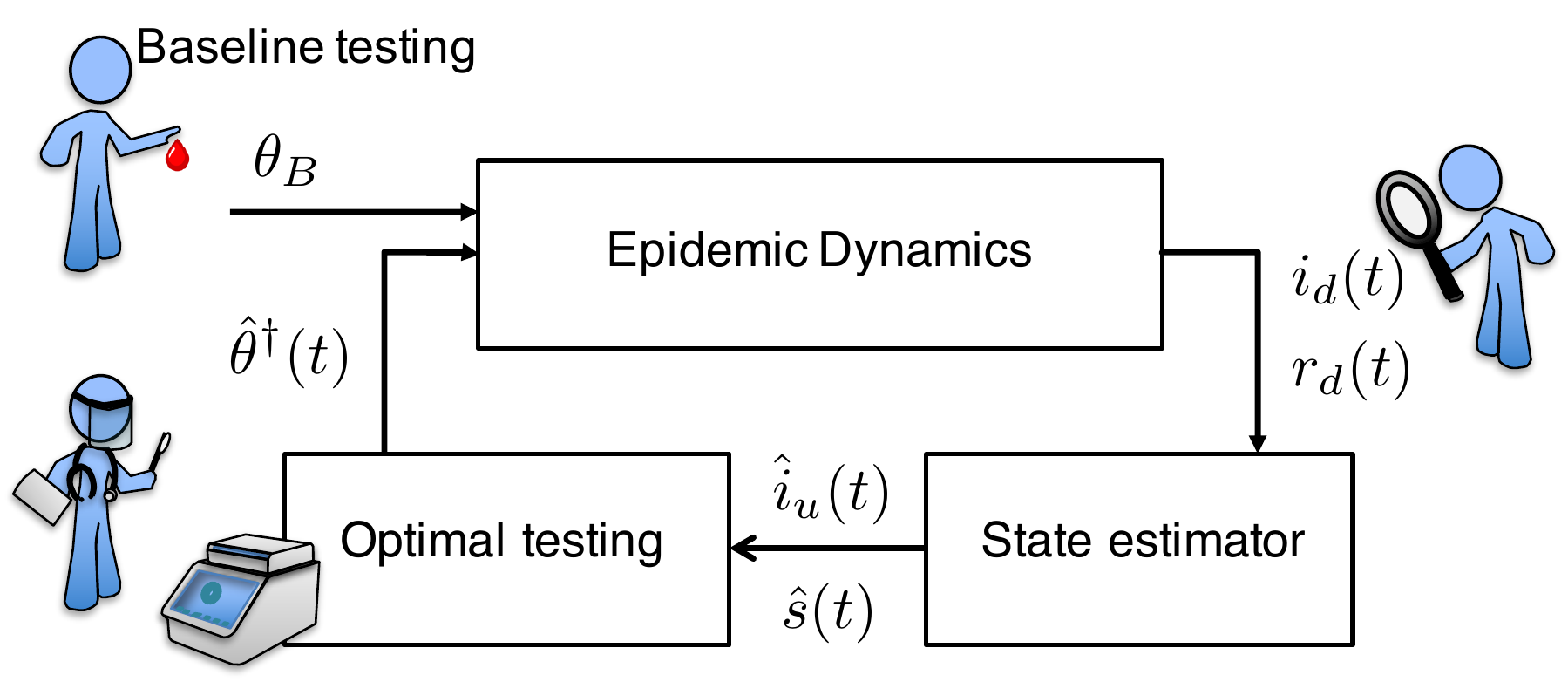}
\caption{Schematic of the proposed procedure. Measurements of detected-infected and
detected-recovered individuals ($i_d(t),r_d(t)$) as defined in Eq.\ \eqref{eq:SIR_estimation} obtained via baseline testing
with rate $\protect\theta_B$ are used to estimate the aggregate
epidemic state ($\hat s(t), \hat i_u(t)$), which is then used to compute
the optimal adaptive testing rate ($\hat \theta^\dagger(t)$%
) according to the results of Theorem~\protect\ref{optimal} (the hat symbol denotes the fact that the optimal testing rate is evaluated as a function of the estimated state). The objective
of adaptive testing is to contain the fraction of infected-undetected
individuals below the desired threshold $i_{\max}$.}
\label{fig:scheme}
\end{figure}

We next argue that state estimation may be infeasible with molecular testing,
which is highly  sensitive but detects infections only during a short
window of time. Intuitively, 
detection of current infections from molecular testing is not  always sufficient to
identify whether a given trajectory of infections is due to a particular time
path of $\beta (t)$ coupled with a given set of initial conditions, or  to a
different time path of transmission coupled with a different set of initial
conditions.  In contrast, such identification is always possible with baseline serology testing.

To illustrate these points, we consider a more detailed model where we partition the
recovered population into recovered-undetected ($\textup{R}_\textup{u}$), consisting of
individuals who had the disease but are not recorded as having immunity
(because they were not diagnosed with either test) and recovered-detected ($%
\textup{R}_\textup{d}$), which consists of individuals that are known to have immunity
because they were either detected during their illness or at a later time
(via serology testing).
Correspondingly, we consider an augmented model that includes both
adaptive and baseline testing with different sensitivities (Fig. \ref%
{fig:sketch_SIR}B). All individuals are tested via adaptive molecular testing
with rate $\theta(t)$ and sensitivity $\eta$, and via baseline testing with
sensitivity $\eta_{BI}$ for the detection of current infections and $%
\eta_{BR}$ for the detection of past infections. The model equations read: 
\begin{equation}  \label{eq:SIR_estimation}
\begin{aligned}
\frac{d{s}(t)}{dt} &= -\beta(t) s(t) i_u(t) \\
\frac{d{i}_u(t)}{dt} &=
\beta(t) s(t) i_u(t) - \gamma i_u(t) -\eta \theta(t) i_u(t) - \theta_B
\eta_{BI} i_u(t)\\
\frac{d{i}_d(t)}{dt} &= \eta \theta (t) i_u(t) +\theta_B
\eta_{BI} i_u(t) - \gamma i_d(t)\\
\frac{d{r}_u(t)}{dt} &= \gamma i_u(t) - \theta_B
\eta_{BR} r_u(t)\\
\frac{d{r}_d(t)}{dt} &= \gamma i_d(t) + \theta_B \eta_{BR} r_u(t),
\end{aligned}
\end{equation}
where the last equation is redundant since $r_d=1-s-i_u-i_d-r_u$.

 We use the more detailed model in Eq.\ \eqref{eq:SIR_estimation} to study the
system's observability, that is, the question of whether an outside observer or
policymaker can estimate the underlying state from detected cases. We prove
that even when $i_{d}(t)$ is observed perfectly and continuously in time
but there is no detection of recovered individuals, the underlying state cannot always be estimated  (see Lemma \ref{lem:obs1} in the Materials and Methods Section \ref{sec:observability}).
Instead, we prove that 
when serology testing is used, which  allows the correct reconstruction of the time
path of both $i_{d}(t)$ and $r_{d}(t)$ via observations of both ongoing and past infections, the underlying state can always be estimated (see Lemma \ref%
{lem:obs2} in  the Materials and Methods  Section \ref{sec:observability}). This
result does not depend on the frequency of baseline testing, nor on the exact
sensitivity of serology testing.

\subsection{Extensions}\label{sec:ext}

In the previous sections we presented a rigorous analysis of optimal testing
and observability for a simple yet insightful deterministic SIR model. Such
results are derived under the assumption of perfect and continuous time
observations. We next discuss some extensions to account for non-idealities
encountered in practice.

First, we consider a model where individuals are detected not only via the
testing program but also because they may become symptomatic. Specifically,
we assume that infected-undetected individuals may develop symptoms and thus
become infected-detected with rate $\kappa $ (this
leads to an additional flow from infected-undetected to infected-detected
with rate $\kappa$ as detailed in  Eq.\ \eqref{eq:SIR_estimationE} in the
Supplementary Materials).
Accordingly, we consider a variant of the original optimization problem
where the constraint is imposed on the \textit{total} number of infected
individuals, instead of infected-undetected individuals only. We also impose
a constraint on the maximum testing rate, modeling daily limitations in
processing capacity. Overall, this results in the following extended optimal
control problem:
\begin{equation}  \label{eq_full_problem}
\begin{aligned} \min_{\theta(\cdot)\ge 0} &\quad
\int_{t=0}^{\infty}\theta(t) dt \\ \text{such that:} & \quad i_u(t)+i_d(t)
\leq i_{\max} \quad \forall t \\ &\quad \theta(t) \leq \theta_{\max} \quad
\forall t. \end{aligned}
\end{equation}
To find the optimal adaptive testing policy $\theta^{*}(t)$ for the
extended problem of Eq. \eqref{eq_full_problem}, we adopted a numerical approach
using the interior point optimizer library within the GEKKO optimization
suite \cite{beal2018gekko,wachter2006implementation}. The
optimal testing policy, computed numerically using parameters taken from the
literature on the COVID-19 epidemic (see Materials and
Methods Table \ref{table:parameters} and Section \ref{sec:parameters}) is shown in Fig. \ref{fig:Optimal_SIR2}. 
\begin{figure}[h]
\centering
\includegraphics{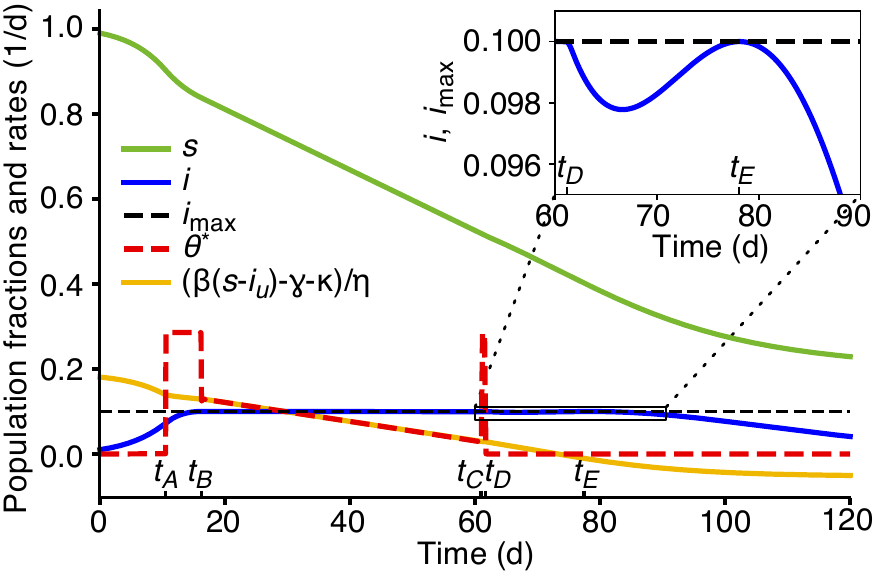}
\caption{Optimal testing strategy for the extended problem \eqref{eq_full_problem}. The green and
blue curves are, respectively, the fractions of susceptible and infected
individuals in the population, respectively. The black, dashed line represents the
constraint on the total fraction of infected individuals, $i=i_u+i_d\leq i_{\max}$. The
optimal testing policy (red, dashed curve) is equal to zero at first, switches to its maximum
value $\protect\theta_{\max}$ at time $t_A$, such that $i_u+i_d$ reaches the
constraint $i_{\max}$ with zero derivative at time $t_B$. Between times $t_B$
and $t_C$, the optimal testing policy is equal to $\protect(\beta %
(s(t)-i_u(t)) - \protect\gamma-\kappa)/\eta$ (yellow curve), which keeps $i_u+i_d=i_{\max}$. At time $%
t_C $, the optimal testing policy switches back to $\protect\theta_{\max}$
until time $t_D$, after which it is equal to zero. The times $t_C$ and $t_D$
are such that, after $t_D$, the total fraction of infected individuals grows
initially, reaching the constraint $i_{\max}$ tangentially (inset), and then
decreases to zero. The switching times can be computed analytically, given
the initial condition (as discussed in the Supplementary Materials).  For illustration purposes, we set $i_{\max}=0.1$ and $\eta=1$, while the other parameters are as in Table \ref{table:parameters}.}
\label{fig:Optimal_SIR2}
\end{figure}
Remarkably, the optimal adaptive testing policy for the extended
optimization problem of Eq.\ \eqref{eq_full_problem} follows the same
principle as the optimal testing policy for the original optimization problem of Eq.\
\eqref{eq_reduced_problem}, that is, it aims at keeping the constrained
quantity ($i_u$ for Eq.\ \eqref{eq_reduced_problem} and $i_u+i_d$ for Eq.\
\eqref{eq_full_problem}) at the threshold $i_{\max}$ for as long as
possible, in order to bring the epidemic as fast as possible to a point
after which it naturally goes to extinction. Two differences arise in
 the testing policy that optimizes the modified  problem  \eqref{eq_full_problem}.
First, in the original problem, one can afford to delay testing right until
the time at which $i_u$ reaches the constraint $i_{\max}$. This is possible
because in the original model the first derivative of the constrained
quantity, $di_u/dt= i_u(\beta s - \gamma -\eta \theta)$, depends explicitly
on the testing frequency $\theta$, and thus one can set the testing rate to $%
\theta=(\beta s - \gamma)/\eta$ to instantaneously ensure $d i_u/dt=0$. In the
extended model, instead, the first derivative of the constrained quantity $%
i_u+i_d$ does not depend directly on the testing frequency $\theta$.
Therefore, one cannot instantaneously impose $d(i_u+i_d)/dt=0$ and thus
testing must start before the constrained quantity $i_u+i_d$ reaches the
threshold $i_{\max}$. The optimal testing policy thus switches from $%
\theta=0$ to $\theta=\theta_{\max}$ at a time $t_A$ (for which $%
i_u(t_A)+i_d(t_A)<i_{\max}$) such that $i_u+i_d$ reaches $i_{\max}$ with zero
derivative at time $t_B>t_A$ (Fig. \ref{fig:Optimal_SIR2}). After time $t_B$, the optimal testing strategy switches to a frequency that maintains $i_u+i_d$ at the $i_{\max}$ value
(the specific rate can be computed from equating the second derivative of $%
i_u+i_d$ to zero). 
Second, $i_u+i_d$ can naturally decrease even before herd immunity if $i_d>0$. For this reason, towards the end of the epidemic the optimal testing
policy for the extended problem of Eq.\ \eqref{eq_full_problem} adopts a second
phase with maximal testing frequency $\theta_{\max}$ that decreases $i_u$
and increases $i_d$ up to a point when, if testing is stopped, $i_u+i_d$
naturally remains below the threshold $i_{\max}$ for all subsequent times
(it increases initially but again reaches the threshold tangentially, see inset of Fig. \ref{fig:Optimal_SIR2}). 
The switching times can be characterized analytically as discussed in the
Supplementary Materials.

Next, we allow the dynamics of the epidemic to be governed by a stochastic
process rather than a deterministic (albeit time-varying) one. This is, of course, more realistic given the stochastic
nature of transmissions and the time-varying and stochastic transition
of individuals across different compartments. Additionally, we assume that observation of  detected cases happens at discrete time instants (e.g. daily) instead of continuously. To deal with this type of
non-idealities and stochasticity  in our state estimation, we propose the use of a
state-estimator which, given observations at
discrete time instants $t_{k}$, produces estimates of the state of the
system (denoted by $\hat{s},\ \hat{\imath}_{u},\ \hat{\imath}_{d},\ \hat{r}%
_{u},\ \hat{r}_{d}$), which can then be used to implement the adaptive
testing policy. 
 For the purpose of
this analysis, we assume that $\beta $ is known and constant, and use an
extended Kalman filter with state constraints as state estimator (Materials
and Methods Section \ref{sec:kalman})  coupled with a system size expansion of the master equation governing the probability distribution of the stochastic model to derive the dependence of
process noise on the epidemic state and population size (Materials and Methods Section \ref{sec:stochastic_model}). Extensions to unknown and time-varying $\beta $ are discussed
in the Supplementary Materials Section \ref{sec:beta_var}. To validate our procedure in the presence of
non-idealities, we used a receding horizon implementation $\hat{\theta}^*$ of the optimal
testing policy derived for the deterministic SIR model (see Materials and Methods Section \ref{sec:receding}). 
Fig. \ref{fig:kalman}A shows the performance of the state estimator and of the
adaptive testing policy for multiple stochastic realizations with $%
\theta _{B}=1/14$ d$^{-1}$ (where d stands for day). The extended Kalman filter provides good
estimates (black, dashed curves in Fig. \ref{fig:kalman}C-D) of the real
state of the epidemic (blue and orange curves), which lies within the
confidence bounds of the estimate. In addition, the time-varying testing
rate implemented using the estimated state is effective in maintaining the
number of infected individuals around the desired threshold $I_{\max
}=Ni_{\max }$ (where $N$ is the population
size). Fluctuations of order $\sqrt{N}$ around such threshold are to be expected as the epidemic is simulated
as a stochastic, Markov process (Materials and Methods Section  \ref{sec:stochastic_model}). Finally, we show that the
 mean of the receding horizon testing policy, computed across realizations of the stochastic model of epidemics, follows the
optimal policy derived for the deterministic SIR model (Fig. \ref{fig:kalman}B).

\begin{figure}[]
\centering
\includegraphics{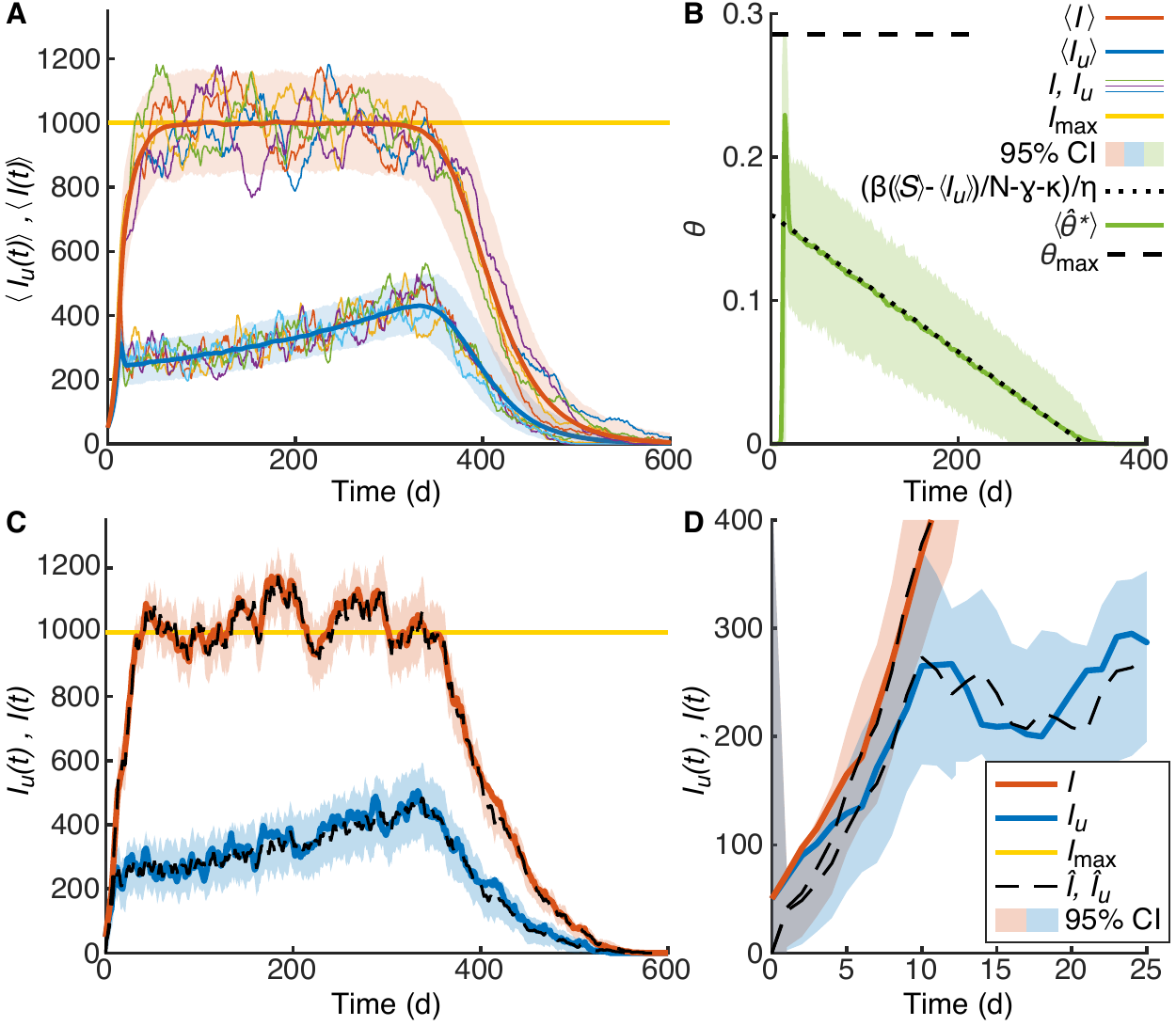}
\caption{Control of stochastic trajectories by using a receding horizon
version of the optimal testing policy $\protect\theta^*$ in combination with
baseline serology testing with rate $\protect\theta_B=1/(14 \textup{\ days}%
) $. Panel A shows the mean number of infected-undetected individuals $%
\langle I_u(t) \rangle$ (blue, thick curves), the mean total number of
infected individuals $\langle I(t) \rangle$ (orange, thick curves) across $%
500$ realizations and $I_u(t)$, $I(t)$ in five, randomly selected stochastic
trajectories (solid, thin lines). Note that capital letters denote absolute numbers instead of fractions of individuals. Colored bands are $95\%$ empirical
confidence intervals and the thick, yellow lines show the value of $I_{\max}$%
. Panel B shows the mean molecular time-varying testing rate in the simulations
(thick, green curve) and its $95\%$ confidence interval. The dotted, black
line shows $(\protect\beta (\langle S \rangle-\langle I_u\rangle)/N-\protect\gamma-\protect%
\kappa)/\eta$, which is the functional form of the optimal adaptive testing
rate $\protect\theta^{*}(t)$ for the deterministic SIR model in the interval [$%
t_B,t_C$] (Eq.\ \protect\ref{optimalBC}). The black, dashed line shows the
maximum testing rate $\protect\theta_{\max}$. Panel C shows the number of
infected-undetected individuals $I_u(t)$ (blue curve) and the total number
of infected individuals $I=I_u(t)+I_d(t)$ (orange curve) in a single
realization. The estimated number of infected-undetected individuals $\hat
I_u(t)$ and estimated total number of infected individuals $\hat I(t)$ are
shown with black, dashed curves. The 95$\%$ confidence intervals for $\hat
I_u(t)$ and $\hat I(t)$ computed using the predicted variance estimated
according to the extended Kalman filter are shown as colored bands. The
infected threshold value $I_{\max}$ is shown as a yellow line. Panel D shows
a zoom of the initial phases of the epidemic highlighting the accuracy of
the Kalman filter estimates. Model parameters and initial conditions are as in the Materials and Methods Table \ref{table:parameters} and Section \ref{sec:parameters}.}
\label{fig:kalman}
\end{figure}

\section{Discussion}

A major lesson from the recent COVID-19 crisis is that, in the absence of
comprehensive vaccines and therapeutic solutions, rapid testing and
isolation become crucial tools to contain the spread of a pandemic. In this paper, we
developed an approach to determine an optimal testing strategy, relevant
especially when there are scarce or expensive testing resources.

Our approach has two basic pillars. First, we showed that, in the context of
a classic SIR model, when the epidemic state in terms of infected,
recovered and susceptible individuals is known and the objective can be
formulated as keeping the number of undetected infections below a certain
threshold, then the optimal testing strategy takes a simple form, similar to
a most rapid approach path. In particular, there should be no testing until
the aforementioned threshold is reached, and thereafter, testing resources
should be used to keep infections at this threshold until herd
immunity is reached and the epidemic starts disappearing naturally. The standard
molecular tests, which have high accuracy, are crucial for this result, because
they enable the identification and isolation of infected individuals.

The second pillar of our approach turns to the identification of the
epidemic state. Our optimal testing policy crucially depends on such
knowledge, but where does this knowledge come from? We tackle this question
by adopting a state estimation framework, where the underlying state is
unknown but can be estimated from the sequence of infections and additional
information obtained from testing. Though molecular tests are also useful in this
context (because they reveal the  trajectory of infections), our main result in
this part is that this information by itself is not sufficient for
identifying the underlying state. This is because the dynamics of 
the SIR model are highly nonlinear and dependent on initial conditions, and
it is not always possible to tell apart whether a given sequence of
infections is due to one of many time-varying paths of transmission rates
coupled with different initial conditions. Instead, we showed that
serology testing which is lower-accuracy, cheaper and longer-range (in
terms of estimating past infections) can be useful to disentangle this information.

More specifically, we proposed a two-pronged approach in which baseline
serology testing is used to collect information about the state of the
epidemic, and the more costly and sensitive molecular testing is adaptively
deployed based on such information (Fig.~\ref{fig:scheme}). Our analysis
formalizes the notion that serology offers advantages as a baseline testing
tool not only because of cost benefits, but also because it conveys
information about past infections, which proves fundamental to correctly and
timely estimate the state of the epidemic. We then showed that, based on
information about the state of the epidemic, optimal adaptive molecular
testing can be adopted and implemented.

Inevitably for a mathematical analysis based on a stylized model, our
approach simplified many aspects of the problem. Our extensions dealt with
two such aspects. First, we showed that similar insights apply in the
context of an SIR model in which both detected-infected and
undetected-infected numbers have to be kept below a certain threshold.
Second, our state estimation techniques apply even when we are dealing with
a stochastic model of epidemics. 

Our analysis suggests that there are tangible gains from the proposed approach. Fig.~\ref{fig:cost} shows that our two-pillar approach with state
estimation plus optimal testing can lead to significant reductions of
overall cost with respect to constant testing strategies, leading to up to
60\% cost reduction for the parameters we investigated.

\begin{figure}[]
\centering
\includegraphics{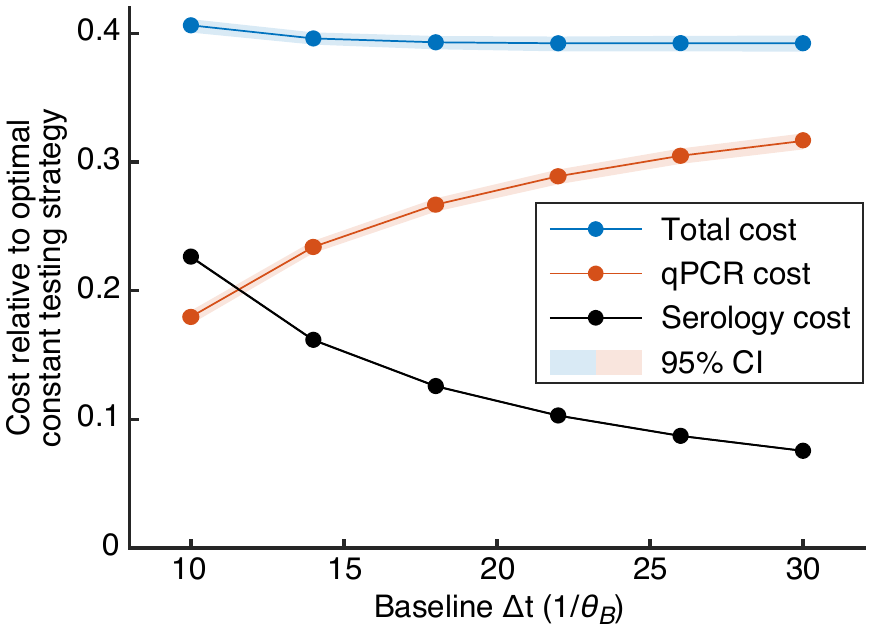}
\caption{Cost of the optimal testing policy for the optimization problem \protect\ref%
{eq_full_problem}, as a function of different baseline testing rates $%
\protect\theta_B$. Costs are normalized with respect to the cost of the
constant testing policy with the minimum testing frequency required to
maintain the constraints $I_u+I_d\leq I_{\max}$ at all times (Materials and
Methods Section \ref{sec:constant_strategy}). Data points connected by straight lines are mean statistics across $1000$ simulations of the stochastic trajectories. Shaded bands represent 95$%
\%$ confidence intervals. Orange and black data points report the 
relative contribution of adaptive qPCR testing and baseline serology testing to the
total cost (blue points), respectively. All adaptive policies induce significant cost
savings, up to 60\% reduction with respect to the cost under the optimal constant
testing strategy derived in Section \ref{sec:constant_strategy} for the parameters investigated. Model parameters and initial conditions are as in the Materials and Methods Table \ref{table:parameters} and Section \ref{sec:parameters}.}
\label{fig:cost}
\end{figure}

In concluding, we make three additional remarks. First, optimality of the proposed
approach is claimed in terms of the problems formulated in Eqs.~\eqref{eq_reduced_problem} and \eqref{eq_full_problem} where the only
epidemic constraint comes from keeping infections below a desired threshold $i_{\max }
$. Strategies that achieve
such an objective are typically classified as \textquotedblleft containment
strategies\textquotedblright , since their objective is to contain the
disease to a state that can be handled by the health system. This is very
different from \textquotedblleft eradication strategies\textquotedblright\
where instead the objective is to eradicate the disease as fast as possible.
It is important to note that the more costly constant testing strategy
detailed in the Materials and Methods Section \ref{sec:constant_strategy} would lead to faster eradication of the
disease than the adaptive strategy suggested here, but testing would
need to continue indefinitely to ensure that outside infections would not
create further waves, as for the parameters considered in the simulations
herd immunity is not reached before eradication under constant
testing. Whether eradication or containment strategies should be preferred
depends on considerations about testing availability, possible long-term
effects of infections on the health of individuals \cite{davis2020longterm} and the impact of the
epidemic on the economy, and is outside of the scope of this work. Our
objective here was to derive the optimal containment strategy for cases when
eradication is simply not possible, e.g., because of test scarcity or budget
limitations.

Second, we derived our results for a standard SIR model (yet with extensions
to testing and symptomatic individuals). Our objective was to derive
analytic insights for a model of epidemics that is general enough to encode
the core traits of an epidemic without getting lost in the details of
specific diseases. Clearly, caution is required when applying our findings
in the field and additional steps are needed to validate our suggestions
with detailed models of any specific disease before translation to practice.
We note especially that we did not account explicitly for delays due to test
processing time in our model. However, the fact that serology has typically
a faster turnaround time than qPCR is an additional argument in support of
serology as a baseline testing tool.

Finally, we investigated the robustness of our procedure to sources of
stochasticity that are intrinsic to the spread of an epidemic, and found that using information on the expected scaling of
fluctuations with the population size and the state of the epidemic, one can
apply the testing strategies developed for the deterministic SIR model to
control epidemics even in the presence of intrinsic stochasticity. In
practical applications, one may encounter additional sources of exogenous
stochasticity due to people's behavioral responses, changes in policies,
infections coming from external sources such as neighboring states, seasonal
changes, etc. We believe that a dual approach where our theoretical results
are used as a guideline for formulating candidate policies that are then
tested extensively with numerical approaches  adopting an ensemble of models (as in \cite{ray2020ensemble,viboud2019future})
would be a powerful tool in the control of future pandemics.

\section{Materials and Methods}

\subsection{Model parameters}\label{sec:parameters1}

In our numerical studies, we used model parameters that have been
used in the literature to describe the spread of COVID-19. Table \ref%
{table:parameters} reports the parameter values, along with the
corresponding sources. Note that we decided to use low sensitivity
both for molecular and serology testing to be conservative and to account for the
fact that infected agents in the initial incubation period may not be
detectable.

{ \begin{table}[htp]
\begin{center}
\begin{tabular}{c|c|c}
Parameter & Value & Sources \\ \hline
Transmission rate ($\beta$) & $0.3$ d$^{-1} $ & \cite{della2020network} \\ 
&  & \cite{gatto2020spread} \\ 
&  & \cite{Bertozzi16732} \\ 
Recovery rate ($\gamma$) & $1/14$ d$^{-1}$ & \cite{della2020network} \\ 
Rate of symptoms development ($\kappa$) & $0.04$ d$^{-1}$ & \cite%
{McAloone039652} \\ 
&  & \cite{gatto2020spread} \\ 
qPCR sensitivity ($\eta$) & $0.9$ & \cite{Watsonm1808} \\ 
Serology sensitivity (current infections) ($\eta_{BI}$) & $0.6$ & \cite%
{roche}, \cite{serology} \\ 
&  & \cite{publichealthengland} \\ 
Serology sensitivity (past infections) ($\eta_{BR}$) & $0.8$ & \cite{roche}, 
\cite{serology} \\ 
&  & \cite{publichealthengland} \\ 
Cost of serology relative to qPCR  ($c^{ser})$ & 0.4 & %
\cite{forbes1} \\ 
&& \cite{forbes2} \\ 
Maximum adaptive testing rate ($\theta_{\max}$) & $2/7$ d$^{-1}$ &  \cite{cornell}
\end{tabular}%
\end{center}
\caption{Parameter values and corresponding sources. Our estimate for $%
\protect\beta$ is a compromise between different estimates reported in the
literature on COVID-19. The rate of symptoms development $\protect\kappa$
was estimated as the product of the probability of becoming symptomatic,
times the incubation rate. }
\label{table:parameters}
\end{table}}

\subsection{Constant testing strategy}\label{sec:constant_strategy}
 One possibility for controlling an epidemic with a constant testing rate is selecting a value of the testing rate $ \theta_{const}$ that guarantees a basic reproduction number \cite{daley2001epidemic} smaller than unity. For  the model {with symptomatic agents
(Supplementary Materials Eq.~\ref{eq:SIR_estimationE}) }  the basic reproduction number is 
$$R_0= \frac{\beta}{\gamma+\kappa+ \eta \theta_{const}},$$
corresponding to the number of secondary infections generated by an individual when he/she is free to circulate and the rest of the population is made entirely of susceptible individuals. Setting the basic reproduction number to unity and solving for   $\theta_{const}$ leads to a testing rate  guaranteeing that the number of infected undetected is monotonically decreasing. This is a more restrictive condition than what is needed to satisfy the constraint $i_u+i_d\le i_{\max}$ in Problem \eqref{eq_full_problem}. Indeed,  a lower testing rate for which the fraction of infected-undetected initially increases but reaches the constraint $i_{\max}$ tangentially would  suffice. To be  fair in the comparison with the adaptive testing policy, we next derive such lower constant testing rate as a function of the fractions 
$s_0$ and $i_0$ of susceptible and infected individuals (all assumed to be
undetected) at time $t=0$. Note that to satisfy the constraint with the least amount of constant testing, $\theta_{const}$ should be such that $i=i_u+i_d$ reaches the constraint $%
i_{\max}$ tangentially (i.e., only once at time $\bar t$ with first derivative equal to zero).
Our objective is to derive  a series or relations between the initial state ($s(0),i_u(0), i_d(0)$), the state at time $\bar t$ (i.e., $s(\bar t),i_u(\bar t), i_d(\bar t)$), and the testing rate $\theta_{const}$. We  then exploit these relations to solve for $\theta_{const}$. From the discussion above we have \begin{equation}  \label{eq:tangent1} i_u(\bar t)+i_d(\bar t)=i_{\max}\end{equation} and 
\begin{equation}  \label{eq:tangent}
\frac{d i(\bar t)}{dt}=\beta s(\bar t) i_u(\bar t)-\gamma i(\bar t)=\beta s(\bar t)
i_u(\bar t)-\gamma i_{\max}=0.
\end{equation}
Next, by Eq.\ \eqref{eq:SIR_optimal} it holds ${ds(t)}/{dt} = -\beta s(t) i_u(t)$ and $ \textstyle{di_d(t)}/{dt}+ \textstyle{dr(t)}/{dt} = (\eta \theta_{const}   + \gamma) i_u(t) $ leading to 
\begin{align*}
 \frac{1}{s(t)} \frac{ds(t)}{dt} +  \frac{\beta}{(\eta \theta_{const}   + \gamma)} \frac{d(i_d(t)+r(t))}{dt} =0.
 \end{align*}
Integrating this equation we can derive a constant of motion for the epidemic which leads to the following relation between the epidemic state at time zero and at time $\bar t$:
\begin{equation}  \label{eq:harko}
 \ln\left(\frac{s(\bar t)}{s_0}\right)-%
\frac{\beta}{\gamma+\eta\theta_{const}}\left(s(\bar t)+i_u(\bar t)-s_0-i_0\right)=0.
\end{equation}
Finally, integrating $d i_d(t) /dt= \eta\theta_{const} i_u(t)-\gamma
i_d(t) $ in the interval $[0,\bar t]$ we obtain 
\begin{equation}  \label{eq:idot}
i_d(\bar t)= \eta\theta_{const} e^{-\gamma  \Delta t(s(\bar t),s_0,i_0) } \int_{s(\bar t)}^{s_0} 
\frac{ e^{\gamma \Delta t(s,s_0,i_0)} }{\beta s}
ds=i_{\max}-i_u(\bar t),
\end{equation}
where we used $i_d(0)=0$ and a reformulation  $\Delta t(\tilde s,s_0,i_0)$ of the
time interval such that $s(\Delta t)=\tilde s$  as a function of the epidemic state as introduced in \cite%
{harko2014} and detailed in Supplementary Materials Eq.\
\eqref{g}.  Solving Eqs.\
\eqref{eq:tangent1}-\eqref{eq:idot}  for the unknowns $s(\bar t), i_u(\bar t), i_d(\bar t), \theta_{const}$
leads to the minimum constant testing rate $ \theta_{const}$ that ensures $%
i(t)=i_u(t)+i_d(t)\leq i_{\max}$ for all $t$.

\subsection{Observability notions}\label{sec:observability}

We formally define observability for a parametric system as follows.

\begin{definition}
A dynamical system $dx(t)/dt=g(x(t),\beta(t))$ with state $x(t)$ and
time-varying parameter $\beta(t)$ is observable from the output $%
y(t)=h(x(t)) $ if for any two observed outputs $y_1(t)$ and $y_2(t)$, the
condition $y_1(t)\equiv y_2(t)$ for all $t$ implies $x_1(0)=x_2(0)$ and $%
\beta_1(t)\equiv \beta_2(t)$ for all $t$.
\end{definition}

\begin{lemma}[Observability from molecular testing]
\label{lem:obs1} Consider the system of Eq.\ \eqref{eq:SIR_optimal} with
state $x(t)=[s(t),i_u(t),i_d(t),r(t)]$ and continuous time output $%
y(t)=i_d(t)$. Suppose that $\eta \theta(t)>0$ for all $t$ and that $\gamma$ is
known.

\begin{enumerate}
\item If $\beta(t)\equiv \beta$, the system is observable.

\item If $\beta(t)$ is time-varying, the system is not observable.
\end{enumerate}
\end{lemma}

\begin{lemma}[Observability from serology testing]
\label{lem:obs2} Consider the system of Eq.\ \eqref{eq:SIR_estimation} with
state $x(t)=[s(t),i_u(t),i_d(t),r_u(t),r_d(t)]$ and continuous time output $y(t)=%
\left[i_d(t),r_d(t)\right]$. If $\eta \theta(t)+\eta_{BR}\theta_B>0$ and $%
\gamma$ is known, the system is observable.
\end{lemma}

The proofs of these lemmas are provided in the Supplementary Materials.

\subsection{Stochastic model}\label{sec:stochastic_model}

We performed stochastic simulations of a compartmental model of epidemics in
which individuals of a population of size $N$ are assigned to the same
compartments $\textup{S}$, $\textup{I}_\textup{u}$, $\textup{I}_\textup{d}$, $\textup{R}_\textup{u}$ and $\textup{R}_\textup{d}$ as in the deterministic SIR model
with symptomatic individuals (Supplementary Materials Eq.\
\eqref{eq:SIR_estimationE}). In analogy with the deterministic SIR model,
transition rates among states are set to: 
\begin{equation*}
\begin{aligned} W(S-1,I_u+1,I_d,R_u|S,I_u,I_d,R_u) &=\beta \frac{SI_u}{N} &
{\text{(new infection)}}\\
W(S,I_u-1,I_d,R_u+1|S,I_u,I_d,R_u)&=\gamma I_u&
{\text{(recovery of $\text{I}_\text{u}$)}}\\
W(S,I_u-1,I_d+1,R_u|S,I_u,I_d,R_u)&=\left[
\eta \theta+ \kappa +\theta_B \eta_{BI} \right] I_u& {\text{(detection
of $\text{I}_\text{u}$)}}\\ W(S,I_u,I_d-1,R_u|S,I_u,I_d,R_u)&=\gamma I_d& {\text{(recovery
of $\text{I}_\text{d}$)}}\\ W(S,I_u,I_d,R_u-1|S,I_u,I_d,R_u)&= \theta_B \eta_{BR} R_u&
{\text{(detection of $\text{R}_\text{u}$)}} \end{aligned}
\end{equation*}
where $W(X^{\prime }|X)$ is the probability per unit time of transitioning
from state $X$ to state $X^{\prime }$ and the parameters have the same
interpretation as in the deterministic SIR model, and capital letters $S$, $I_u$%
, $I_d$ and $R_u$ indicate the absolute number of individuals in the various
compartments of the stochastic model. The compartment $\textup{R}_\textup{d}$ is not mentioned explicitly, as its
abundance is equal to $N-S-I_u-I_d-R_u$. Unlike the deterministic SIR model, the
stochastic model of epidemics accounts for the fact that the numbers of
individuals in each compartment are integers and that infection, recovery
and detection are stochastic events. As such, the stochastic model is better
suited to describing epidemics in small populations or  the epidemiological dynamics in the initial phases of an epidemic, where number fluctuations can be important. The dynamics of
the stochastic model is governed by the master equation, \cite{vankampen,gardiner2009stochastic}: 
\begin{equation}
\begin{aligned} \frac{\partial P}{\partial t}(S,I_u,I_d,R_u,t) =&\bigg(
\left( \mathbf{E}^{+1}_S \mathbf{E}^{-1}_{I_u}-1 \right) \beta \frac {SI_u}N
+ \left( \mathbf{E}^{+1}_{I_u} \mathbf{E}^{-1}_{R_u}-1 \right) \gamma I_u +
\\ &+ \left( \mathbf{E}^{+1}_{I_u} \mathbf{E}^{-1}_{I_d}-1 \right) I_u
\left( \eta \theta+\kappa + \theta_B \eta_{BI} \right) + \left(
\mathbf{E}^{+1}_{I_d}-1\right) \gamma I_d+\\ &+ \left(
\mathbf{E}_{R_u}^{+1}-1 \right) \theta_B \eta_{BR} R_u \bigg)
P(S,I_u,I_d,R_u,t), \end{aligned}  \label{eq:MASTER}
\end{equation}
where $P(S,I_u,I_d,R_u,t)$ is the probability of being in state $X:=[S,I_u,I_d,R]$ at time $t$ and the transition operator $\mathbf{E}^{+1}_S$
is defined by $\mathbf{E}^{\pm1}_S f(S,I_u,I_d,R_u)=f(S\pm1,I_u,I_d,R_u)$
for a generic function $f$, and similarly for the other operators.\footnote{In Eq.\
\eqref{eq:MASTER}, operators within parentheses act on all the functions of
state variables to their right according to conventional operator
precedence, e.g. $\left( \mathbf{E}^{+1}_S \mathbf{E}^{-1}_{I_u}-1 \right) {%
SI_u} P(S,I_u,I_d,R_u,t)$ is to be interpreted as $\left( \mathbf{E}^{+1}_S 
\mathbf{E}^{-1}_{I_u}-\!1\! \right) {SI_u} P(S,I_u,I_d,R_u,t)\!=\! \mathbf{E}^{+1}_S 
\mathbf{E}^{-1}_{I_u} \!\left( {SI_u} P(S,I_u,I_d,R_u,t)\right)\!-\! {SI_u}
P(S,I_u,I_d,R_u,t)$ $= (S+1)(I_u-1)P(S+1,I_u-1,I_d,R_u,t)- {SI_u}
P(S,I_u,I_d,R_u,t)$.} We simulated trajectories of the stochastic model of epidemics by using the
Gillespie algorithm \cite{gillespie1976general}, with the parameters
reported in Table~\ref{table:parameters}.

For large $N$, Eq. \eqref{eq:MASTER} can be expanded in
powers of  $1/{N}$ following a Kramers-Moyal or system-size expansion \cite{vankampen,gardiner2009stochastic}. Eq. \eqref{eq:MASTER} can be expressed in terms of the rescaled variables $\tilde x:=X/N=[\s,\iu,\id,\ru]$ as follows:
\begin{equation}
\begin{aligned} \frac 1N \frac{\partial p}{\partial t}(\s,\iu,\id,\ru,t) =&\bigg(
\left( \mathbf{E}^{+\frac 1N}_{\s} \mathbf{E}^{-\frac 1N}_{\iu}-1 \right) \beta \s \iu
+ \left( \mathbf{E}^{+\frac 1N}_{\iu} \mathbf{E}^{-\frac 1N}_{\ru}-1 \right) \gamma \iu +
\\ &+ \left( \mathbf{E}^{+\frac 1N}_{\iu} \mathbf{E}^{-\frac 1N}_{\id}-1 \right) \iu
\left( \eta \theta+\kappa + \theta_B \eta_{BI} \right) + \left(
\mathbf{E}^{+\frac 1N}_{\id}-1\right) \gamma \id+\\ &+ \left(
\mathbf{E}_{\ru}^{+\frac 1N}-1 \right) \theta_B \eta_{BR} \ru \bigg)
p(\s,\iu,\id,\ru,t). \end{aligned}  \label{eq:MASTER_rescaled}
\end{equation}
Note that $\tilde x(t)$ is a stochastic process, whereas $x(t)$ as defined in the main text is the solution to the deterministic SIR model. The  right hand side of Eq.~\eqref{eq:MASTER_rescaled} is a function of $\tilde x\pm1/N$. Expanding this function around $\tilde x$  up to the second order ($1/N^2$), one obtains the Fokker-Planck equation:
\begin{equation}\label{eq:fokker}
\frac{\partial p}{\partial t}(\tilde x,t)=-\sum_j \frac{\partial}{\partial \tilde x_j} \left( g_j(\tilde x,\theta) p(\tilde x,t) \right) + \frac 1{2N} \sum_{j,k} \frac{\partial^2}{\partial \tilde x_j \partial \tilde x_k} \left( B_{jk}(\tilde x,\theta) p(\tilde x,t) \right),
\end{equation}
where $g$ is the vector field corresponding to the deterministic dynamics used for computing the optimal testing strategy (see Eq. \eqref{eq_full_problem} and Eq. \eqref{eq:SIR_estimationE} in the Supplementary Materials) and $B$ is the matrix:
$$\small B(\x,\theta)=\left[\begin{array}{cccc}
\beta \s \iu & -\beta \s \iu & 0 & 0 \\
-\beta \s \iu & \left[ \beta \s+\gamma+\eta\theta+\kappa+\theta_B \eta_{BI} \right]\iu & - \left[ \eta \theta+\kappa + \theta_B \eta_{BI} \right]\iu & -\gamma \iu \\
0 & - \left[ \eta \theta+\kappa + \theta_B \eta_{BI} \right]\iu &\left[ \eta \theta+\kappa + \theta_B \eta_{BI} \right] \iu +
\gamma \id & 0 \\
0 & -\gamma \iu & 0 & \gamma \iu+\theta_B\eta_{BR} \ru\end{array}\right].$$
Eq. \eqref{eq:fokker} is known as the diffusion approximation of the master Eq. \eqref{eq:MASTER_rescaled} and describes the probability distribution of a continuous stochastic process specified by the following It\^{o} Langevin equation: \cite{vankampen,gardiner2009stochastic}
\begin{equation}\label{eq:kalman_dynamics}
\frac{d \tilde x}{d t}=g(\tilde x,\theta)+\frac 1{\sqrt{N}}\varepsilon(\tilde x,\theta,t),
\end{equation}
where $\varepsilon(\tilde x,\theta,t)$ is a Gaussian noise with covariance $\langle \varepsilon_j(\tilde x(t),\theta(t),t) \varepsilon_k(\tilde x(t'),\theta(t'),t') \rangle=B_{jk}(\tilde x(t),\theta(t))\delta(t-t')$ and zero mean, where $\delta$ is the Dirac delta. Intuitively, in Eq. \eqref{eq:kalman_dynamics} the first term $g(\tilde x,\theta)$ coincides with the vector field of  the deterministic dynamics while the second  term captures diffusive fluctuations due to stochasticity, whose amplitude depends both on the population size and on the epidemic state.
As detailed in the next section, this approach enables us to characterize the process noise for the extended Kalman filter  with the correct scaling in $N$ and highlights its dependence on the current epidemic state (e.g., the number of infected-undetected) as captured by  the matrix $B$.

\subsection{State estimation for the stochastic simulations}
\label{sec:kalman}

Given the dynamics of Eq. \eqref{eq:kalman_dynamics}, we discuss here how we estimate the state of the epidemic. We  assume the following observation model:

\begin{equation}\label{eq:output}
\tilde y(t_k)=[ \tilde i_d(t_k) ; \tilde r_d(t_k)]= C \tilde x(t_k) +c
\end{equation}
with: 
\begin{equation*}
C=\left[%
\begin{array}{cccc}
0 & 0 & 0 & 1 \\ 
-1 & -1 & -1 & -1%
\end{array}%
\right],\qquad c= \left[%
\begin{array}{c}
0 \\ 
q%
\end{array}%
\right],
\end{equation*}
 that is, only infected-detected and recovered-detected are observed, and we assume discrete observation times $t_k$ (e.g., daily observations). Such
observations can be used to estimate the state of the system via a state
observer, which we implemented using an extended Kalman filter with state
constraints (see \cite{king2008inapparent,pasetto2017real,pasetto2018near,li2020substantial} for other applications of the Kalman filter in the context of epidemiology).

In a Kalman filter, observations $y_k=\tilde y(t_k)$ are used to create an estimate of
the state, denoted by $\hat x(t) =[\hat s(t); \hat i_u(t); \hat i_d(t) ;\hat
r_u(t)]$. The first step is to initialize $\hat x_{0\mid 0}=\mathbb{E}[x(t_0)], P_{0\mid 0}=\mathbb{E}%
[(x(t_0)-\hat x(t_0)) (x(t_0)-\hat x(t_0))^\top]$. Then, the dynamics of the
extended Kalman filter \cite{khalilbook} is computed as follows, at any time step $t_k$:
\begin{enumerate}
\item Predict the next state, given previous observations:
\begin{equation*}
\left\{\begin{aligned} \frac{d\hat{x}(t)}{dt}&=g(\hat{x}(t), \theta(t))\\ \frac
{dP(t)}{dt}&= G(t)P(t)+P(t)G(t)^\top+Q(t) \end{aligned}\right. \ \textup{with}
\ \left\{ \begin{aligned} \hat x(t_{k})&=\hat x_{k\mid k}\\
P(t_{k})&=P_{k\mid k} \end{aligned}\right. \textup{\ and } G(t)=\left. \frac{%
\partial g}{\partial x}\right|_{\hat x(t),\theta(t)}
\end{equation*}
and set $\hat x_{k+1\mid k}=\hat x(t_{k+1}), P_{k+1\mid k}=P(t_k)$
\item Update the prediction, given the current observation:
\begin{align*}
K_{k+1}&=P_{k+1\mid k} C^\top (C P_{k+1\mid k} C^\top + R)^{-1} \\
\hat x_{k+1\mid k+1}&=\Pi_{X_{k+1}}[\hat x_{k+1\mid k} + K_{k+1} (y_{k+1} - C
\hat x_{k+1\mid k})] \\
P_{k+1\mid k+1}&=(I - K_{k+1}C)P_{k+1\mid k},
\end{align*}
where $\Pi_{X_k}$ represents the projection in the feasible set $X_k=\{x\geq
0\mid Cx+c=y_k\}$.
\end{enumerate}

The matrices $Q$ and $R$ are covariance matrices for
the process and measurement noise. In our analysis, we assume $R=0$ (as the number
of infected-detected and recovered-detected is perfectly know by the policy
maker and thus there is no measurement error in Eq. \eqref{eq:output}), while $Q(t)$ is the covariance of the process noise, which is equal to $B/N$ as derived from the expansion of the master equation in Section \ref{sec:stochastic_model}.

\subsection{Testing strategy for the stochastic simulations}\label{sec:receding}

The testing policy derived for the deterministic SIR model is not necessarily
robust to the presence of stochastic fluctuations. For this reason, in the
stochastic simulations we implemented a receding horizon version $%
\hat{\theta}^*(t)$ of the testing policy where at any time $t_k>t_A$ (as
defined in Fig. \ref{fig:Optimal_SIR2})  we computed the constant testing rate needed to drive the total fraction of infected to the threshold $i_{\max}$ in a
horizon of $H$ days (we set $H=3$ d),
assuming that the dynamics follows the deterministic SIR model, i.e. Eq.\ \eqref{eq:SIR_estimationE} in the Supplementary Materials. According to the principles of receding horizon control, such testing rate is applied for one time step
and then a new problem is solved for the next horizon $[t_{k+1},t_{k+1}+H]$
given the new realized state. 
Thus, at every time step $t_k$ the testing rate is set to: 
\begin{equation}  \label{control2}
\hat{\theta}^*(t_k) =%
\begin{cases}
0 & \textup{if } t_k < t_A \\ 
\max (\{0,\min\{ \theta_{\max}, \theta_{rh}(\hat s(t_k),\hat i_u(t_k) \} \}) & \textup{otherwise}%
\end{cases}%
\end{equation}
where $\theta_{rh}$ is the testing rate that would bring the deterministic
system to the constraint $i_u+i_d=i_{\max}$ with zero derivative in a time
horizon $H$, starting from $\hat s(t_k)$ and $\hat i_u(t_k)$ (see Eq.\ \eqref{theta_rec} in the Supplementary
Materials). Note that we assume here that the receding horizon is implemented for any time $t_k>t_A$, where $t_A$ is the optimal time to start testing as computed for the deterministic SIR model. In practice, the policy maker may prefer to implement the receding horizon control from the beginning, for additional robustness and to compensate for the uncertainty of state estimates in the early phases of the dynamics. This has a minor cost implications, since $t_A$ is typically very small with the parameters considered here.

\subsection{Parameters used in simulations of the stochastic model of epidemics} \label{sec:parameters}

Simulations of the stochastic model of epidemics were performed with population size $N=50000$ and constraint $I_{\max}=1000$ corresponding to $2\%$ of the population size. This percentage was chosen for illustration purposes and it roughly corresponds to the peak percentage quarantine capacity estimated to be required for the safe reopening of Cornell's Ithaca NY campus during the COVID-19 pandemic in the Fall 2020 \cite{cornell}. Realizations of stochastic epidemics were initialized with $I(0)=I_u(0)=50$ infected and $S(0)=N-50$ susceptible individuals (the other compartments were initialized at $I_d(0)=R_d(0)=R_u(0)=0$). The other parameters were set to the values in Table \ref{table:parameters}. The initial state estimate for the extended Kalman filter was set to $\hat I(0)=0$ infected and $\hat S(0)=N$ susceptible individuals (the estimates for the other compartments were set to zero). All entries of the initial estimate for the covariance matrix $P$ (see Section \ref{sec:kalman}) of the extended Kalman filter were set to zero, with the exception of the estimate for the variance of $\hat S(0)$ and of $\hat I_u(0)$, which were set to $I_{\max}^2/12$, to reflect a large uncertainty on the initial condition.

\bibliographystyle{apalike}
\bibliography{main_edited}

\clearpage

\pagenumbering{arabic}
\renewcommand{\thepage}{S-\arabic{page}}

\appendix
\setcounter{equation}{0}
\setcounter{figure}{0}
\renewcommand{\thesection}{S}
\renewcommand{\theequation}{S\arabic{equation}}
\renewcommand{\thefigure}{S\arabic{figure}}

\section*{Supplementary Materials}

\subsection{Optimal testing policy}

We next present the proof of the following theorem. Theorem \ref{optimal} in
the main text is obtained as a special case when the transmission rate is
constant.

\begin{theorem}\label{theorem2}
Suppose that the transmission rate $\beta(t)$ is a monotone non-increasing
function of time. The optimal testing policy $\theta^{\dagger}(t)$ for the
optimization problem of Eq.\ \eqref{eq_reduced_problem} with dynamics as in
 Eq.\ \eqref{eq:SIR_optimal} acts in three phases:

\begin{enumerate}
\item While $i_u(t)<i_{\max}$, do not test, i.e. set $\theta^\dagger(t)=0$.

\item After $i_u(t)$ reaches $i_{\max}$, test with time-varying rate $
\theta^\dagger(t)=({\beta(t) s(t)-\gamma})/{\eta}$.

\item Once herd immunity is reached, that is $s(t)=\gamma/\beta(t) $ for the
first time, stop testing, i.e. set $\theta^\dagger(t)=0$.
\end{enumerate}
\end{theorem}

To prove such theorem we start with a series of auxiliary lemmas. Let $%
t_{herd}$ be the first instant of time such that $s(t_{herd})=\gamma/%
\beta(t_{herd})$. We say that the system reached herd immunity at time $%
t_{herd}$ since for any $t>t_{herd}$ the fraction of infected naturally
decreases. Mathematically,
\begin{equation*}
\frac{d i_u(t)}{dt} = (\beta(t) s(t) -\gamma)i_u(t) \le (\beta(t_{herd}) s(t_{\text{%
herd}}) -\gamma)i_u(t) =0,
\end{equation*}
where we use the fact that in this model the susceptible and $\beta$ are
monotonically non-increasing. We start by proving that after herd immunity
it is optimal to stop testing.

\begin{lemma}
\label{lemma:after_herd_immunity} Under the optimal testing policy $%
\theta^\dagger$ there exists a finite time $t_{herd}$ at which $%
s(t_{herd})=\gamma/\beta(t_{herd})$. Moreover, $\theta^\dagger(t) = 0$ for
all $t\ge t_{herd}$. Here the superscript $\dagger$ denotes the evolution
under the optimal testing policy.
\end{lemma}

\begin{proof}
The fact that under the optimal testing policy herd immunity is reached is immediate as if that was not the case the objective function would be infinite. Let $t_{herd}$ be the time when herd immunity is reached and recall that $\beta(t_{herd}) s(t_{herd}) - \gamma \leq 0$ for all $t \geq t_{herd}$ since $\beta(t)$ and $s(t)$ are both decreasing with time. Consequently, $d i_u(t)/dt\le 0$ for all $t \geq t_{herd}$  and the control policy that sets $\theta^\dagger(t) = 0 \ \forall \ t \geq t_{herd}$ is feasible and therefore optimal (as no other control policy can achieve zero cost).
\end{proof}
We next show that, under the optimal control, once herd immunity is reached
the number of infected undetected must be at the threshold.

\begin{lemma}
\label{lemma:herd_imm_imax} If the optimal objective is strictly positive, $i^\dagger_u(t_{herd}) =
i_{\max}$.
\end{lemma}

\begin{proof}
First, note that if the optimal objective is $0$, this means that we do not apply any control to the system. The optimal trajectory $i_u^\dagger(t)$ is therefore at most tangent to $i_{\max}$, because otherwise, we would need to apply some positive control to make sure that it does not violate the constraint $i_u \leq i_{\max}$. Instead, if the optimal objective is strictly positive there is an interval of time before herd immunity is reached with positive control. 
If $i^\dagger_u(t_{herd}) < i_{\max}$, we could decrease the control by a small amount before reaching herd immunity (the last time the control was positive before reaching herd immunity). This would imply that $i_u$ increase faster but for small deviations of the control  we could still guarantee $i_u(t)\le i_{\max}$ for all times. Since the rate of decrease of $s$ increases as $i_u$ increases, this would imply that $s$ will decrease faster and therefore herd immunity will be reached at a time $t_{herd}' < t_{herd}$. From Lemma \ref{lemma:after_herd_immunity}, we have that the optimal control after reaching herd immunity is identically 0. Therefore, the new control policy strictly reduces the objective,  violating the optimality of the control.
\end{proof}

Finally, we characterize the optimality of the proposed control $%
\theta^\dagger$ before $t_{herd}$.

\begin{lemma}
The optimal control takes the {most rapid approach path} to reach
herd immunity.
\end{lemma}

\begin{proof}

Lemma \ref{lemma:after_herd_immunity}  characterizes how the optimal control behaves after herd immunity. Therefore, we can rewrite  the original optimization problem in Eq.  \eqref{eq_reduced_problem} as follows
\begin{align}\label{eq5}
\min & \int_{t=0}^{t_{herd}} \theta(t) dt \nonumber \\
\text{s.t. } \quad  & i_u(t) \leq i_{\max} \quad \forall t  \nonumber \\
& \theta(t) \geq 0 \quad \forall t
\end{align}
where $t_{herd}$ is the time to reach herd immunity (which depends on $\theta$). Integrating the  dynamics 
\begin{align*}
\eta \theta(t) = \beta(t) s(t) - \gamma -\frac{1}{i_u(t)} \frac{d{i_u(t)}}{dt}
\end{align*}
yields
\begin{align*}
\eta  \int_{t=0}^{t_{herd}} \theta(t) dt &=  \int_{t=0}^{t_{herd}} \left( \beta s(t) - \gamma  -\frac{1}{i_u(t)} \frac{di_u(t)}{dt} \right) dt \nonumber \\
 &=  \int_{t=0}^{t_{herd}} \left( \beta s(t) - \gamma  \right) dt - \int_{t=0}^{t_{herd}} \left( \frac{1}{i_u(t)} \frac{di_u(t)}{dt} \right) dt \nonumber \\
 &=  \int_{t=0}^{t_{herd}} \left( \beta s(t) - \gamma \right) dt + \log \left( \frac{i_u(0)}{i_u({t_{herd})}} \right) \nonumber \\
 &=  \int_{t=0}^{t_{herd}} \left( \beta s(t) - \gamma \right) dt + \log \left( \frac{i_u(0)}{i_{\max}} \right),
\end{align*}
where we have used Lemma \ref{lemma:herd_imm_imax} for the last equality. Note that the second term is a constant and $\eta$ is a positive constant. Therefore, the original optimal control problem can be rewritten as:
\begin{align}\label{eq55}
\min & \int_{t=0}^{t_{herd}} \left( \beta s(t) - \gamma  \right) dt \nonumber \\
\text{s.t. } \quad  & i_u(t) \leq i_{\max} \quad \forall t  \nonumber \\
& \theta(t) \geq 0 \quad \forall t.
\end{align}

We next show that the  trajectory $i_u^\dagger$ induced by the control $\theta^\dagger$ pointwise dominates any other feasible trajectory  $i_u$ and therefore the corresponding trajectory  $s^\dagger$ is pointwise smaller than any other feasible trajectory of $s$, before hitting herd immunity  (Fig. \ref{fig:Proof_Image}). This has two consequences: i) herd immunity is reached sooner under $\theta^\dagger$ and ii) $ \left( \beta(t) s(t) - \gamma  \right)$ is pointwise smaller. These two points prove that   $\theta^\dagger$ minimizes Eq. \eqref{eq55} and thus also Eq. \eqref{eq5}.

\noindent To prove points i) and ii) recall that from the first line of Eq.\ \eqref{eq:SIR_optimal}
\begin{align*}
\frac{ds}{dt} = -\beta s i_u \qquad \implies
\frac{1}{s}\frac{ds}{dt} = -\beta i_u \qquad \implies \frac{d\log(s)}{dt} = -\beta i_u
\end{align*}
and by integration
\begin{align*}
\log \left( \frac{s(t)}{s(0)}\right) = - \int_0^t \beta(\tau) i_u(\tau) d\tau
\end{align*}
which shows that for two feasible trajectories $(s^\dagger(t), i^\dagger_u(t))$ and $(s(t),i_u(t))$, if $i_u^\dagger(t) \geq i(t)$ for all $t \leq T$, we have that $s^\dagger(t) \leq s(t)$ for all $t \leq T$.
\end{proof}

\begin{figure}[H]
\centering
\includegraphics{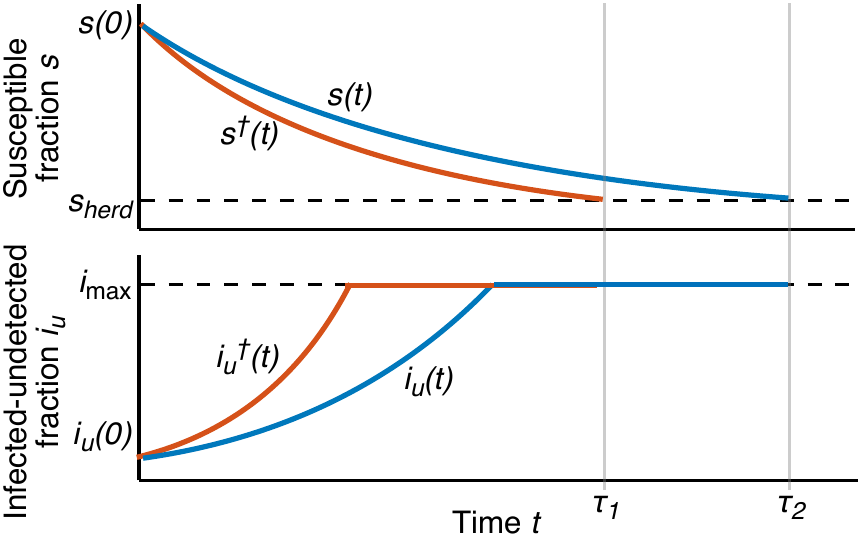}
\caption{The trajectory $i_u^\dagger(t)$ induced by the optimal testing strategy $\theta^\dagger$ pointwise dominates any other feasible trajectory $i_u$. Conversely, the trajectory $s^\dagger(t)$ induced by the optimal testing strategy $\theta^\dagger$ is pointwise smaller than any other feasible trajectory $s(t)$, until herd immunity is reached. In this illustration $\beta$ is constant, so that $s(t_{herd})=\gamma/\beta=s_{herd}$.}
\label{fig:Proof_Image}
\end{figure}

Combining the lemmas above proves Theorem \ref{theorem2}. The expression of $%
\theta^\dagger$ when $i_u(t)=i_{\max}$ can be obtained from ${d i_u}/{dt}=(\beta
s-\gamma -\eta \theta^\dagger) i_u=0$.

\subsection{Proofs of Observability}

\subsubsection{Proof of Lemma \protect\ref{lem:obs1}}

We prove the two statements separately.
\begin{enumerate}
\item From the third line of Eq.\ \eqref{eq:SIR_optimal}, it holds 
\begin{equation*}
i_u(t)=\frac 1{\eta\theta(t)} \left( \frac{di_d(t)}{dt}+\gamma i_d(t) \right)
\end{equation*}
hence $i_u$ (and all its derivatives) can be reconstructed from the observed
output $i_d$ (and its derivative). From the second line of Eq.\ \eqref{eq:SIR_optimal}, for $\beta$ constant, we obtain 
\begin{equation}  \label{eq:step}
\beta s(t) =\frac{1}{i_u(t)}\frac{di_u(t)}{dt} + \gamma +\eta \theta(t).
\end{equation}
Substituting the expression in Eq. \eqref{eq:step} on the right hand side of 
\begin{equation*}
\beta = -\frac{1}{\beta s(t) i^2_u(t)} \left( \frac{d^2i_u(t)}{dt^2}-\beta s(t)\frac{di_u(t)}{dt} +\gamma \frac{di_u(t)}{dt} +
\eta \frac{d\theta(t)}{dt}i_u(t) +\eta\theta(t) \frac{di_u(t)}{dt} \right)
\end{equation*}
(obtained by computing the second derivative of $i_d$) yields a formula for $%
\beta$ as a function of known quantities ($i_u$, $i_d$, $\theta$ and their
derivatives). Since $\beta s(t)$ is known from Eq. \eqref{eq:step}, this implies
that $s(t)$ is known and finally $r(t)=1-s(t)-i_u(t)-i_d(t)$.

\item If $\beta(t)$ is time-varying the system is not observable. To show
this we consider two evolutions that start from different initial conditions
and have different $\beta$ evolutions, yet lead to the same observable
output $i_d$. This proves that just observing the output it is not possible
to distinguish the two scenarios. Specifically consider the two systems:

\begin{align*}
&\left\{
\begin{aligned} \frac{ds}{dt} &= -\beta s i_u \\
\frac{di_u}{dt} &= \beta s i_u - \gamma
i_u -\eta\theta(t) i_u \\
\frac{di_d}{dt} &=\eta\theta(t) i_u - \gamma i_d
\end{aligned}\right. & \left\{\begin{aligned} \frac{d\bar s}{dt} &= -\bar \beta \bar
s\bar i_u \\ \frac{d\bar i_u}{dt} &= \bar \beta \bar s \bar i_u -
\gamma\bar i_u -\eta\theta(t) \bar i_u \\
\frac{d\bar i_d}{dt} &=
\eta\theta(t) \bar i_u - \gamma \bar i_d \end{aligned}\right.
\end{align*}
with initial state $s(0)\neq \bar s(0),$ $i_u(0)= \bar i_u(0)$, $i_d(0)=
\bar i_d(0)=r(0)= \bar r(0)=0$ and suppose that $\bar
\beta(t)=\beta(t) {s(t)}/{\bar s(t)}$. Then 
\begin{align*}
\frac{d i_u}{dt} &= \beta s i_u - \gamma i_u -\eta\theta(t) i_u , \\
\frac{d\bar i_u}{dt} &= \beta s \bar i_u - \gamma\bar i_u -\eta\theta(t)
\bar i_u ,\quad \bar i_u(0)= i_u(0)
\end{align*}
Since $i_u$ and $\bar i_u$ solve the same differential equation it must be 
$\bar i_u(t)\equiv i_u(t)$ for all $t$. This immediately implies $%
i_d(t)\equiv \bar i_d(t)$, yet the evolution of $s$ and $\bar s$ is
different as they start from different initial conditions.
\end{enumerate}

\subsubsection{Proof of Lemma \protect\ref{lem:obs2}}

Since $y(t)$ is observed continuously in time one can use it to compute
exactly $di_d/dt$ and $d r_d/dt$. The third and fifth equations in %
Eq. \eqref{eq:SIR_estimation} can then be used to recover $i_u(t)$ and $r(t)$
exactly as follows: 
\begin{align*}
i_u(t) &= \frac1{\eta\theta(t) + \theta_B \eta_{BI}} \left(\frac{di_d}{dt} + \gamma i_d\right),
\notag \\
r_u(t) &= \frac 1{\theta_B \eta_{BR}} \left(\frac{dr_d}{dt} - \gamma i_d\right).
\end{align*}
Using the fact that 
\begin{align*}
s(t) = 1 - i_u(t) - r_u(t) - i_d(t) - r_d(t)
\end{align*}
one can reconstruct $s(t)$ as well. Overall, the state can be estimated
exactly from the observed variables. As a byproduct,  knowledge of $s(t)$
and $i_u(t)$ allows the identification of $\beta(t)$ from the first equation
in Eq. \eqref{eq:SIR_estimation}.

\subsection{Optimal testing policy for the extended problem in Eq.\ \protect
\ref{eq_full_problem}}

We consider an extension of the model in Eq.\ \eqref{eq:SIR_estimation} that
accounts for detection of symptomatic individuals by considering an additional flow from infected-undetected to infected-detected with rate $\kappa$: 
\begin{equation}  \label{eq:SIR_estimationE}
\begin{aligned} \frac{ds(t)}{dt} &= -\beta(t) s(t) i_u(t) \\
 \frac{di_u(t)}{dt}&=
\beta(t) s(t) i_u(t) - \gamma i_u(t) -\eta \theta(t) i_u(t) -\kappa i_u(t) -
\theta_B \eta_{BI} i_u(t)\\
 \frac{di_d(t)}{dt} &= \eta \theta (t) i_u(t) + \kappa
i_u(t)+\theta_B \eta_{BI} i_u(t) - \gamma i_d(t)\\
 \frac{dr_u(t)}{dt}&= \gamma
i_u(t) - \theta_B \eta_{BR} r_u(t)\\
 \frac{dr_d(t)}{dt} &= \gamma i_d(t) +
\theta_B \eta_{BR} r_u(t) \end{aligned}
\end{equation}

The numerical solution of the problem in Eq.\
\eqref{eq_full_problem} for this extended model (and for the parameters of Fig. \ref{fig:Optimal_SIR2}) has the structure described in  Section~\ref{sec:ext} and schematized in Fig. \ref{fig:testing_structure}. To generalize this analysis to any set of parameters, we here aim at deriving an analytic characterization of the optimal testing policy for Problem \ref{eq_full_problem} within the class of policies with such a structure. More in detail, we aim at deriving analytic expressions for the optimal switching times  $t_A$, $t_B$, $t_C$ and $t_D$ and for the testing rate in the interval $[t_B,t_C]$ as a function of the model parameters and initial conditions.
With slight abuse of notation, we  denote the optimal policy within this class with the  symbol $\theta^*$. 
For this analysis we assume $\beta$ to be constant,
and without loss of generality we assume $\eta=1$. Moreover, we work under
the assumption that the state is known, hence we set $\theta_B=0$.

\begin{figure}[]
\centering
\includegraphics{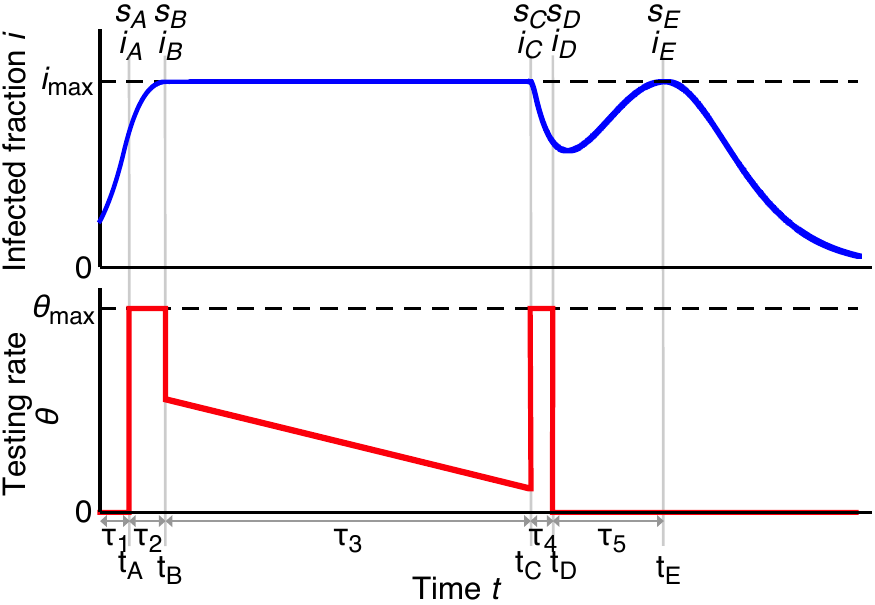}
\caption{Structure of the optimal testing strategy (red curve) and corresponding
effect on the total fraction of infected individuals (blue curve). The optimal
testing policy $\protect\hat{\theta}^*$ is zero until time $t_A$, when it
switches to the maximum testing rate $\protect\theta_{\max}$ until time $t_B$%
, when $i=i_u+i_d$ reaches the constraint $i_{\max}$ with zero derivative.
In the interval $[t_B,t_C]$ of duration $\protect\tau_3$, the optimal
testing policy maintains $i=i_{\max}$, until switching back to $\protect%
\theta_{\max}$ between times $t_C$ and $t_D$. After $t_D$, $i$ increases
until reaching $i_{\max}$ with zero derivative at time $t_E$, and decreases
afterwards.}
\label{fig:testing_structure}
\end{figure}
The following analytic relationships (adapted from \cite{harko2014}) between 
$s$ and $i_u$ at two times $t_1<t_2$ under constant testing rate $\theta$
will be useful: 
\begin{equation}  \label{f}
f_\theta(s(t_1),i_u(t_1),s(t_2),i_u(t_2)) := \ln\left(\frac{s(t_2)}{s(t_1)}%
\right)-\frac{\beta}{\gamma+\theta+\kappa}%
\left(s(t_2)+i_u(t_2)-s(t_1)-i_u(t_1)\right)=0,
\end{equation}
\begin{equation}\begin{aligned}  \label{g}
t_2 - t_1 &= \int_{e^{- \frac{\beta}{\gamma+\theta+\kappa}
(1-s(t_1)-i_u(t_1))}}^{\frac {s(t_2)}{s(t_1)} e^{ -\frac{\beta}{%
\gamma+\theta+\kappa} (1-s(t_1)-i_u(t_1))}} \frac{1}{x} \frac{1}{%
-\beta-(\gamma+\theta+\kappa) \ln x + \beta s(t_1) x e^{ \frac{\beta}{%
\gamma+\theta+\kappa}(1-s(t_1)-i_u(t_1))}}dx \\
&:=\Delta t_{\theta}(s(t_2),s(t_1),i_u(t_1)).
\end{aligned}\end{equation}

\subsubsection{Intervals $[0,t_A]$ and $[t_A,t_B]$}

\label{sec:intervals0AB} First, we evaluate $s_A:= s(t_A)$, $i_{uA}:=
i_u(t_A)$, $s_B:= s(t_B)$ and $i_{uB}:= i_u(t_B)$ as functions of $s_0$ and $%
i_0$. Using Eq.\ \eqref{f} we have: 
\begin{align}
f_0(s_0,i_0,s_A,i_{uA})=0 ,  \label{A1} \\
f_{\theta_{\max}}(s_A,i_{uA},s_B,i_{uB})=0 .  \label{A2}
\end{align}
Imposing that at time $t_B$ one has $(i_u+i_d)\mid_{t_B}=i_{\max}$ and ${%
d(i_u+i_d)/dt}\mid_{t_B}=0$ yields: 
\begin{equation}
\left. \frac{d(i_u+i_d)}{dt}\right|_{t_B}=\beta s_Bi_{uB}-\gamma (i_{uB}+i_d(t_B))=\beta
s_Bi_{uB}-\gamma i_{\max}=0 ,  \label{A3}
\end{equation}
and integrating $d i_d /dt= (\theta_{\max}+\kappa) i_u -\gamma i_d$ in $%
[t_A,t_B]$ one finds: 
\begin{equation}
\begin{aligned} i_{dB}&=i_{\max}-i_{uB}= \left(
i_{dA}+(\theta_{\max}+\kappa) \int_{t_A}^{t_B} e^{\gamma (t-t_A)} i_u(t) dt
\right) e^{-\gamma (t_B-t_A)}\\ &= \left( i_{dA}+ (\theta_{\max}+\kappa)
\int_{s_A}^{s_B} e^{\gamma \Delta t_{\theta_{\max}}(s,s_A,i_{uA})}i_u(s)
\frac{dt}{ds}ds \right) e^{-\gamma \tau_2} \\ &= \left(
i_{dA}+(\theta_{\max}+\kappa) \int_{s_B}^{s_A} \frac{ e^{\gamma \Delta
t_{\theta_{\max}}(s,s_A,i_{uA})} }{\beta s} ds \label{A4}\right) e^{-\gamma
\tau_2} , \end{aligned}
\end{equation}
where $\tau_2$ can be computed from Eq.\ \ref{g} as a function of $s_A$, $%
i_{uA}$ and $s_B$. The quantity $i_{dA}$ can be computed as a function of $%
s_A$ by integrating $d i_{d}/dt=-\gamma i_d+\kappa i_u$ in the interval $%
[0,t_A]$, which gives: 
\begin{equation*}
i_{dA}=e^{-\gamma t_A}\int_{s_A}^{s_0} e^{\gamma \Delta t_{0}(s,s_0,0)} 
\frac{\kappa}{\beta s}ds,
\end{equation*}
where again we made use of Eq.\ \ref{g}. Eqs. \eqref{A1}-\eqref{A4} are four
equations in four unknowns that can be solved to find $s_A,i_{uA},s_B,i_{uB}$
as a function of $s_0$ and $i_0$. In the interval $[t_A,t_B]$, the optimal
testing policy is $\theta^{*}(t)=\theta_{\max}$, and thus the cost of the
control in the interval $[0,t_B]$ is: 
\begin{equation*}
\int_0^{t_B} \theta^{*}(t)dt = \tau_2 \theta_{\max}.
\end{equation*}

\subsubsection{Interval $[t_B,t_C]$}

In the interval $[t_B,t_C]$ the optimal testing policy maintains $%
i=i_u+i_d\equiv i_{\max}$. Hence, the first and second derivatives of $i$
must be equal to zero. Note that 
\begin{align*}
\frac{d(i_u+i_d)}{dt}&= \beta s i_u -\gamma (i_u+i_d)=\beta s i_u -\gamma
i_{\max}=0 \\
\frac{d^2(i_u+i_d)}{dt^2}&= \beta \frac{ds}{dt} i_u +\beta s\frac {di_u}{dt} = \beta(-\beta s
i_u) i_u +\beta s(\beta s i_u - \gamma i_u -\theta^{*} i_u - \kappa i_u) \\
&= -\beta^2 s i_u^2 + \beta^2 s^2 i_u -\beta s\gamma i_u -\beta s \theta^{*}
i_u-\beta \kappa s i_u=0
\end{align*}
which yields 
\begin{equation}  \label{optimalBC}
\theta^{*}=\beta(s- i_u) - \gamma-\kappa.
\end{equation}
Substituting this expression in the dynamic equations, we obtain that in the
interval $[t_B,t_C]$ the undetected fraction of the infected individuals
satisfies: 
\begin{equation}  \label{ieq}
\frac{di_u}{dt}=(\beta s -\gamma -\theta^{*} -\kappa )i_u=\beta i_u^2
\end{equation}
and thus 
\begin{equation}  \label{eq:iusol}
i_u(t)=\frac{1}{\displaystyle \frac{1}{i_{uB}}-\beta (t-t_B)}.
\end{equation}
Moreover, 
\begin{equation}\begin{aligned}
\frac{di_d}{dt}&= \theta^{*} i_u +\kappa i_u -\gamma i_d=\beta s i_u -\gamma i_u
-\beta i_u^2-\gamma i_d =-\frac{ds}{dt} -\gamma i_u -\frac{di_u}{dt}-\gamma i_d \\
&\Rightarrow 0= \frac{di_u}{dt}+ \frac{di_d}{dt}= -\frac{d s}{dt}-\gamma(i_u+i_d) \Rightarrow
\frac {ds}{dt}= -\gamma i_{\max} \\
&\Rightarrow s(t)= s_B-\gamma i_{\max} (t-t_B).  \label{eq:ssol}
\end{aligned}\end{equation}
Thus, $s_C=s_B-\gamma i_{\max} \tau_3$ and $i_{uC}={1}/({{1}/{i_{uB}}%
-\beta \tau_3})$. The value of $\tau_3$ is limited from above by the
constraints $\tau_3 <(\beta i_{uB})^{-1}$ (for solvability of Eq.\ \eqref{ieq}%
) and $\theta^{*} \geq 0$, which can be expressed as an upper constraint on $%
\tau_3$ via Eqs. \eqref{optimalBC}, \eqref{eq:iusol} and \eqref{eq:ssol}. We
denote by $\bar \tau_3$ the minimum of such constraints. 
The cost of the optimal
testing policy in the interval $[t_B,t_C]$ is: 
\begin{align*}
\int_{t_B}^{t_C} \theta^{*}(t)dt&= \int_{t_B}^{t_C} \left( \beta(s(t)-
i_u(t)) - \gamma -\kappa \right) dt \\
&=\int_0^{\tau_3}\left( \beta\left( s_B-\gamma i_{\max} t^{\prime }- \frac{1%
}{\frac{1}{i_{uB}}-\beta t^{\prime }} \right) - \gamma - \kappa \right) d
t^{\prime } \\
&= \tau_3 \left(\beta s_B -\gamma -\kappa-\frac{\gamma}{2}\beta
i_{\max}\tau_3 \right) +\ln \left(1-\beta i_{uB}\tau_3 \right).
\end{align*}

\subsubsection{Intervals $[t_C,t_D]$ and $[t_D,t_E]$}

The analysis is identical to that of Section \ref{sec:intervals0AB}, the
objective is to evaluate $s_D$, $i_{uD}$ ,$s_E$ and $i_{uE}$ as a function
of $s_C$ and $i_C$. From Eq. \eqref{f} we have: 
\begin{align*}
f_{\theta_{\max}}(s_C,i_{uC},s_D,i_{uD})=0 \\
f_0(s_D,i_{uD},s_E,i_{uE})=0.
\end{align*}
At time $t_E$, $i_u+i_d$ is tangent to $i_{\max}$, hence: (similar to Eq.\
\eqref{A3}) 
\begin{equation*}
\beta s_E i_{uE}=\gamma i_{\max}.
\end{equation*}
Integrating $d i_d/dt = \theta_{\max} i_u + \kappa i_u -\gamma i_d$ in the
interval $[t_C,t_D]$ yields: 
\begin{equation*}
i_d(D)=\left(i_d(C)+ ( \theta_{\max}+\kappa) \int_{s_D}^{s_C} \frac{
e^{\gamma \Delta t_{\theta_{\max}}(s,s_C,i_C)} }{\beta s} ds\right)e^{-%
\gamma \tau_4},
\end{equation*}
where $i_d(C)=i_{\max}-i_{uC}$. We can also compute $i_d(D)$ from the
equation for $i_d$ in the interval $[t_D,t_E]$ (during which the testing
rate is equal to zero) as: 
\begin{equation*}
i_d(D)=(i_{\max}-i_{uE})e^{\gamma \tau_5}-\kappa \int_{s_E}^{s_D} \frac{%
e^{\gamma \Delta t_{0}(s,s_D,i_D)}}{\beta s}ds,
\end{equation*}
and equating the two expressions for $i_d(D)$ we find: 
\begin{equation*}
(i_{\max}-i_{uE}) e^{\gamma \tau_5}-\kappa \int_{s_E}^{s_D} \frac{e^{\gamma
\Delta t_0(s,s_D,i_D)}}{\beta s}ds=\left(i_{\max}-i_{uC}+ (
\theta_{\max}+\kappa) \int_{s_D}^{s_C} \frac{ e^{\gamma \Delta
t_{\theta_{\max}}(s,s_C,i_C)} }{\beta s} ds\right)e^{-\gamma \tau_4},
\end{equation*}
The time intervals $\tau_4$ and $\tau_5$ can be computed from Eq.\ \eqref{g}.
The cost of the control in the interval $[t_C,+\infty]$ is thus: 
\begin{equation*}
\int_{t_C}^\infty \hat{\theta}^*(t)dt = \tau_4 \theta_{\max}.
\end{equation*}

\subsubsection{Minimization over $\protect\tau_3$ yields the optimal testing
policy}

So far we obtained a formula for the overall cost that, for fixed $s_0$ and $i_0$, depends only on $\tau_3$. Minimizing such cost for $\tau_3\in[0, \bar
\tau_3]$ yields the optimal value for $\tau_3$ (Fig. \ref{fig:costs}), with
which the optimal testing strategy $\theta^{*}$ is characterized. Note that
since $s$ is monotonically decreasing, the testing strategy can be
formulated as follows: 
\begin{equation*}
\theta^{*}(t)=%
\begin{cases}
0 & \textup{if } s(t)>s_A \textup{\ or } s(t)<s_D \\ 
\theta_{\max} & \textup{if } s(t)\in[s_B,s_A]\cup[s_D,s_C] \\ 
\beta(s(t)- i_u(t)) - \gamma-\kappa & \textup{otherwise}%
\end{cases}%
\end{equation*}

\begin{figure}[h]
\centering
\includegraphics{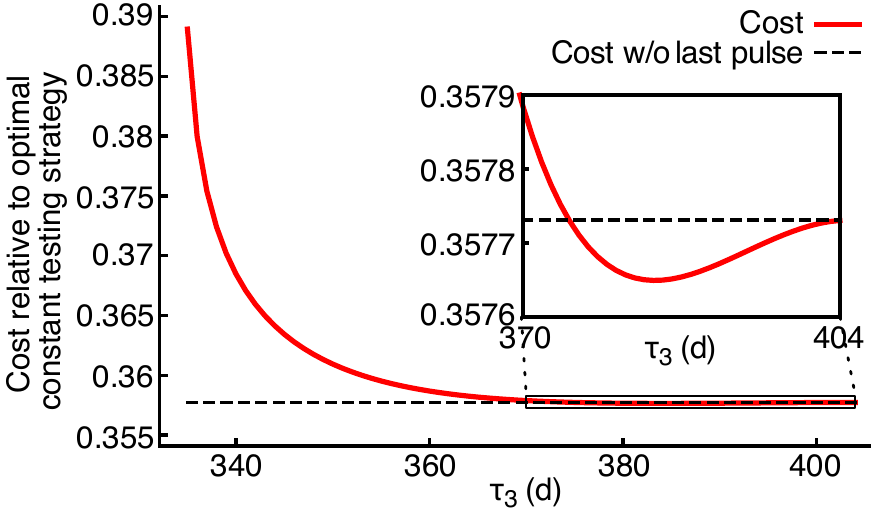}
\caption{Cost of testing policies (red curve) with the structure of Fig. \protect
\ref{fig:testing_structure} for the deterministic SIR model Eq.\ 
\eqref{eq:SIR_estimationE}, as a function of the duration $\protect\tau_3$
of the time interval $[t_B,t_C]$. Costs are computed relative to the cost of
the optimal, constant testing strategy  (Materials and Methods Section \ref{sec:constant_strategy}). The optimal testing policy adopts
the value of $\protect\tau_3$ that minimizes the cost (inset). The dashed,
black line indicates the cost of a testing policy with the largest possible
value of $\protect\tau_3$, i.e. a testing policy in which the last phase of
testing at maximum testing rate $\protect\theta_{\max}$ does not occur. The
inset shows that the benefit of the last phase of testing at maximum testing
rate is marginal.}
\label{fig:costs}
\end{figure}

\subsection{Receding horizon implementation of the optimal testing policy
for the stochastic simulations}

In the receding horizon implementation, at any time $t>t_A$ we aim at
computing the testing rate $\theta_{rh}(t)$ that brings the deterministic dynamics of Eq. \eqref{eq:SIR_estimationE} to $%
i_u+i_d=i_{\max}$ with $d{(i_u+i_d)}/dt=0$ in a time $H$. This leads to the
following set of equations: 
\begin{equation}  \label{control_sim}
\begin{aligned}
&f_{\theta_{rh}(t)}(s(t),i_u(t),s(t+H),i_u(t+H))=0 \\ 
& \left. \frac {d(i_u(t)+i_d(t))}{dt} \right|_{t+H}=\beta s(t+H)i_u(t+H)-\gamma i_{\max}=0 \\ 
&s(t+H)=s(t)-\beta\int_t^{t+H} s(t^{\prime })i_u(t^{\prime })dt^{\prime
}\simeq s(t)- \frac{\beta}2 \left( s(t+H)i_u(t+H)+s(t)i_u(t) \right)H,%
\end{aligned}%
\end{equation}
where $f_\theta$ is as in Eq. \eqref{f} and in the last equation we used the
trapezoidal rule to approximate the integral. We denote the solution of Eq.\
\eqref{control_sim} as: 
\begin{equation}  \label{theta_rec}
\theta_{rh}(s(t),i_u(t)):=-\beta \frac{s(t)+i_u(t)-s(t+H)}{\ln (s(t+H)/s(t))}+\frac{%
\gamma i_{\max}}{s(t+H)\ln(s(t+H)/s(t))}.
\end{equation}

\subsection{Validation with time-varying transmission rate}\label{sec:beta_var}
Fig. \ref{fig:fig_betavar} illustrates the performance of the testing policy 
$\hat{\theta}^*$ for the optimization problem of Eq.\ \ref{eq_full_problem} when
the transmission rate is time-varying (in this case sinusoidal to mimic
seasonality). For any time $t$, the transmission rate $\beta(t)$ is
estimated from the reconstructed state $[\hat{x}(\tau)]_{\tau=t-7}^{t-1} $
over a moving window of length 7 days by nonlinear least square regression.
The testing rate is evaluated by using Eqs. \eqref{control2} and %
\eqref{theta_rec} with the estimated $\hat \beta(t)$ instead of $\beta$. We
see from Fig. \ref{fig:fig_betavar} that the time-varying transmission rate
can be accurately predicted (while the epidemic is active) and the testing
rate implemented using such prediction effectively stabilizes the epidemics
at the desired threshold. This is a preliminary indication that the
suggested testing strategy may be effective even in the presence of an
unknown and time-varying transmission rate. A detailed numerical study is
needed to further support this conclusion and is left as future work.

\begin{figure}[!h]
\centering
\includegraphics[width=0.8\textwidth]{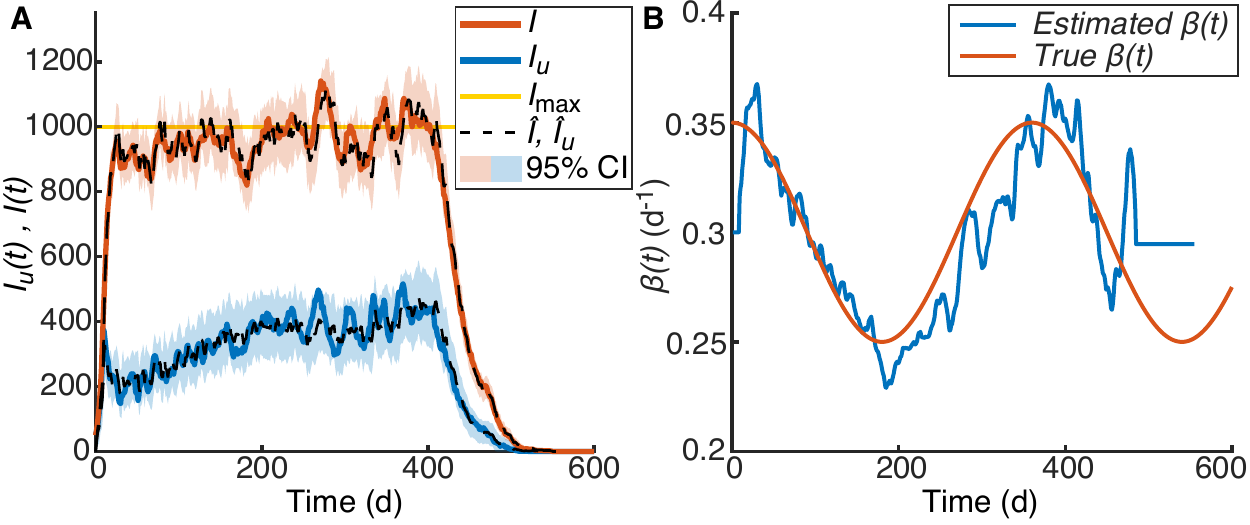}
\caption{Performance of the testing policy $\protect\hat{\theta}^*$ when $%
\protect\beta(t)$ varies sinusoidally with a one year period centered at $%
\protect\beta=0.3$ d$^{-1}$ and amplitude $0.1$ d$^{-1}$. At any time $t_k$, 
$\protect\beta(t_k)$ was estimated using data from the previous $7$ days by
fitting the epidemiological dynamics to the deterministic SIR model, taking into
account the testing rates adopted. Panel A shows the total number of
infected $I$ (orange curve) and infected-undetected (blue curve) in a stochastic
realization of the epidemic, along with the state estimates obtained from
the extended Kalman filter (dashed, black curves). Panel B shows the true
(orange curve) and estimated (blue curve) values of $\protect\beta$.}
\label{fig:fig_betavar}
\end{figure}

\end{document}